\documentclass[letterpaper, 10 pt,journal]{IEEEtran}

\usepackage{graphicx}
\usepackage{amsmath,amssymb}   
\usepackage{tikz}
\usepackage{cite}
\usepackage{subfigure}

\usepackage{amsthm}
\theoremstyle{remark}
\newtheorem{thm}{Theorem}
\newtheorem{defn}[thm]{Definition}
\newtheorem{assu}[thm]{Assumption}
\newtheorem{algo}[thm]{Algorithm}
\newtheorem{lem}[thm]{Lemma}
\newtheorem{lrop}[thm]{Proposition}
\newtheorem{coro}[thm]{Corollary}
\theoremstyle{remark}
\newtheorem{rmk}[thm]{Remark}

\usepackage{mathtools}

\newcounter{MYtempeqncnt}

\title{Data-driven system analysis of nonlinear systems using polynomial approximation (extended version)}
\author{Tim Martin and Frank Allg{\"o}wer
	\thanks{Frank Allgöwer thanks the funding by the Deutsche Forschungsgemeinschaft (DFG, German Research Foundation) under Germany's Excellence Strategy - EXC 2075 - 390740016 and under grant 468094890. We acknowledge the support by the Stuttgart Center for Simulation Science (SimTech). (Corresponding author: Tim Martin).}
	\thanks{The authors are with the University of Stuttgart, Institute for Systems Theory and Automatic Control, 70569, Stuttgart, Germany. (e-mail: tim.martin@ist.uni-stuttgart.de; frank.allgower@ist.uni-stuttgart.de).}
}

\begin{document}

\IEEEoverridecommandlockouts

\IEEEpubid{\begin{minipage}{\textwidth}\ \\[40pt] \copyright 2023 IEEE. Personal use of this material is permitted. Permission from IEEE must be obtained for all other uses, in any current or future media,
		including reprinting/republishing this material for advertising or promotional purposes, creating new collective works, for resale or redistribution to servers
		or lists, or reuse of any copyrighted component of this work in other works.\end{minipage}}

\maketitle
\pagestyle{empty}	


\begin{abstract}
	In the context of data-driven control of nonlinear systems, many approaches lack of rigorous guarantees, call for nonconvex optimization, or require knowledge of a function basis containing the system dynamics. To tackle these drawbacks, we establish a polynomial representation of nonlinear functions based on a polynomial sector by Taylor's theorem and a set-membership for Taylor polynomials. The latter is obtained from finite noisy samples. By incorporating the measurement noise, the error of polynomial approximation, and potentially given prior knowledge on the structure of the system dynamics, we achieve computationally tractable conditions by sum of squares relaxation to verify dissipativity and incremental dissipativity of nonlinear dynamical systems with rigorous guarantees. The framework is extended by combining multiple Taylor polynomial approximations which yields a less conservative piecewise polynomial system representation. The proposed approach is applied for numerical and experimental examples. There it is compared to a least-squares-error model including knowledge from first principle.	  
\end{abstract}


\section{Introduction}\label{Intrduction}

\IEEEPARstart{O}{btaining} a controller for nonlinear systems by nonlinear controller design techniques~\cite{Khalil} usually requires to retrieve a precise model. To this end, first principles can be applied, which is however often time consuming and calls for expert knowledge on the physical phenomena. Alternatively, system identification provides the possibility to derive models of the system from measured trajectories. In case the structure of the system is known, a linearly parametrized model with fixed basis functions reduces the identification to a parameter estimation problem. If the structure is unknown, then an inference on the nonlinear dynamics is more involved. To recover closed-loop stability from controller design techniques with inherent closed-loop guarantees, the model error is required. However, this is an active research field even for the identification of linear time-invariant (LTI) systems \cite{SysIdLin}.\\\indent
Thereby, we are interested in data-driven system analysis where control-theoretic properties are determined from measured system trajectories. This kind of system properties, such as dissipativity \cite{DissiWillems}, provides insight into the system and facilitates a stabilizing controller design by well-investigated feedback laws \cite{Khalil} without knowledge of an explicit model of the system.\\\indent
For LTI systems, \cite{Maupong} and \cite{OneShot} treat the data-driven determination of dissipativity by Willems' fundamental lemma, a trajectory-based non-parametric system representation. \cite{AnneDissi} and \cite{WaardeDissi} determine dissipativity based on data informativity using noisy measurements. While there exist approaches tailored for certain classes of nonlinear systems, comprising methods for general nonlinear systems do not exist. For example, \cite{MartinDissi} generalizes, among others, the results from \cite{AnneDissi} for polynomial systems. For general nonlinear systems, \cite{AnneGaussian} proposes Gaussian process optimization for data-driven system analysis including statistical guarantees. Furthermore, \cite{MontenbruckLipschitz} and \cite{MartinGraphAppro} present an offline and online ansatz for estimating, e.g., the $\mathcal{L}_2$-gain over a finite time horizon exploiting Lipschitz inferences \cite{KinkyInf}. But in fact, this non-parametric characterization for a continuous-time system has intrinsically the drawbacks that a data-based stability analysis of an equilibrium remains denied regardless of the amount of samples. Moreover, an obscure neighbourhood of an equilibrium must be excluded when verifying dissipativity \cite{MontenbruckLipschitz}. \cite{AnneGaussian,MontenbruckLipschitz,MartinGraphAppro} also require thousands of input-output trajectories which prevents their application in practice. Due to the conservatism of the S-procedure, Lipschitz inference is also not suitable for dissipativity verification by linear matrix inequalities (LMI) as indicated in \cite{MartinGraphAppro}.\\\indent
For that reason, we established in\cite{MartinNL} a novel data-based representation for general nonlinear functions relying on Taylor polynomial (TP) approximation. This kind of approximation is commonly exploited in control theory \cite{DifferentialSys} and application, e.g., to linearize a system dynamics. In \cite{MartinNL}, we first determine by Taylor's theorem a polynomial sector for a nonlinear function. In a second step, we formulate a set-membership for the TP of the unknown function from noisy samples. Together with the polynomial sector, this results in a data-based envelope that contains the graph of the ground-truth function. Therefore, this envelope can be leveraged to analyze control-theoretic properties from noise-corrupted input-state-velocity measurements with guarantees. Due to the polynomial characterization of the envelope, the system property verification boils down to a convex semi-definite programming (SDP) via sum of squares (SOS) relaxation \cite{SOSTutorial}. Thus, \cite{MartinNL} shows the verification of dissipativity and its application for two numerical examples. However, the data-based representation is also suitable to determine optimal integral quadratic constraints by adapting the framework of \cite{MartinNLMSOS} and to compute a state-feedback controller with closed-loop stability and performance guarantees \cite{Gaussian_noise}.\\\indent
The contribution of this work is to provide a comprising view on the novel data-based representation by TPs including discussions on its conservatism and the selection of several degrees of freedom. Furthermore, we provide throughout this article necessary results for the successful application of this approach in practice. More specifically, a purely data-based representation for general nonlinear dynamics is often conservative, and therefore jeopardizes a meaningful inference on system properties. For that reason, we extend the ellipsoidal outer approximation from \cite{MartinNLMSOS} and \cite{Ellipsoid_DePersis} by incorporating additionally prior knowledge on the structure of the system dynamics. To further improve the accuracy of the data-driven inference, we combine multiple TP approximations to derive a piecewise polynomial representation similar to a spline approximation. From system identification literature \cite{Poly_Spline}, it is well-known that splines exhibit better approximations than polynomials. Thus, a piecewise polynomial representation significantly improves the inference on system properties. Although this corresponds to a nonlinear approximation, the dissipativity verification can still be executed by SOS optimization. A further contribution of this article is the probabilistic validation of assumptions from data with probabilistic guarantees. Finally, we show the application of the approach for a nonlinear experimental example with state constraints using extensively the mentioned improvements.\\\indent 
While dissipativity \cite{DissiWillems} has played a crucial role in nonlinear system analysis and control \cite{Khalil} since decades, the interests in incremental and differential dissipativity for nonlinear controller design have risen recently, compare \cite{DifferentialSys,IncrementalRoland, ContractionProcess}. Whereas incremental dissipativity properties have not been determined by TP approximation and from data yet, we present their verification by our data-driven representation together with a numerical example. Thereby, we show that our TP representation is not limited to verify dissipativity but can easily be adapted for a large range of control problems.\\\indent
The outline of this article is the following. After introducing some notation in Section~\ref{SecNot}, we elaborate in Section~\ref{SecNPModel} the polynomial data-driven representation of an unknown nonlinear function based on TP approximation. Moreover, Section~\ref{SecNPModel} includes a validation procedure for the essential assumptions. On this basis, Section~\ref{SecDissi} establishes the verification of dissipativity properties for nonlinear dynamical systems by the data-driven inference from a single TP and multiple TPs. Furthermore, this section presents the application for an experimental example. In Section~\ref{SecIncr}, we extend the framework to determine incremental dissipativity and show its application on a numerical example. Section~\ref{SecConclusion} concludes the article. 
\IEEEpubidadjcol

\section{Notation and preliminaries}\label{SecNot}

Let $||v||_2$ be the Euclidean norm of a vector $v\in\mathbb{R}^{n}$, $I_n$ be the $n\times n$ identity matrix, and $0$ be a zero matrix of suitable dimensions. For some matrices $A_1$ and $A_2$ of suitable dimensions, we use the abbreviations $\star^T A_2\cdot A_1=A_1^TA_2A_1$ and $\text{diag}(A_1\big| A_2)=	\begin{bmatrix}\begin{array}{c|c}	A_1 & 0 \\\hline 0 & A_2 	\end{array}	\end{bmatrix}$.\\\indent
Furthermore, let $\mathbb{R}[x]$ be the set of all real polynomials $p$ in $x=\begin{bmatrix}
x_1 & \cdots & x_n\end{bmatrix}^T\in\mathbb{R}^n$, i.e., 
\begin{equation*}
p(x)=\sum_{\alpha\in\mathbb{N}^n,|\alpha|\leq d} a_\alpha x^\alpha,
\end{equation*}
with vectorial indices $\alpha=\begin{bmatrix}\alpha_1 & \cdots & \alpha_n\end{bmatrix}^T\in\mathbb{N}^n$, real coefficients $a_\alpha\in\mathbb{R}$, monomials $x^\alpha=x_1^{\alpha_1}\cdots x_n^{\alpha_n}$, $|\alpha|=\alpha_1+\cdots+\alpha_n$, and the degree $d$ of $p$. Analogously, the set of all $m$-dimensional polynomial vectors is denoted by $\mathbb{R}[x]^m$ and $m\times n$-polynomial matrices by $\mathbb{R}[x]^{m\times n}$, where each entry is an element of $\mathbb{R}[x]$. By convention, the degree of a polynomial vector and matrix is the largest degree of their elements.\\\indent  
A quadratic polynomial matrix $P\in\mathbb{R}[x]^{n\times n}$ with even degree is an SOS matrix, or SOS polynomial for $n=1$, if a matrix $Q\in\mathbb{R}[x]^{{m}\times {n}}$ exists with $P=Q^TQ$. Here we write $\text{SOS}[x]^{{n}\times {n}}$ as the set of all ${n}\times {n}$ SOS matrices. By the square matricial representation \cite{SOSDecomp} of SOS matrices, the verification whether a polynomial matrix is an SOS matrix boils down to an SDP with an LMI constraint. Hence, this verification is computationally tractable. To incorporate the data-based representation in the context of system analysis, a generalized S-procedure will be essential.
\begin{lrop}[SOS relaxation]\label{SOSRelaxation}
	A polynomial $p\in\mathbb{R}[x]$ is not negative for all $x\in\{x\in\mathbb{R}^n:c_1(x)\geq0,\dots,c_k(x)\geq0\}$ with $c_i\in\mathbb{R}[x]$ if there exist SOS polynomials $t_i\in\text{SOS}[x],i=1,\dots,k,$ such that $p-\sum_{i=1}^{k}t_ic_i\in\text{SOS}[x]$.
\end{lrop}
\begin{proof}
	A proof based on the Positivstellensatz can be found in \cite{ProofProp} (Lemma 2.1).
\end{proof}

Since not every non-negative polynomial is an SOS polynomial, Proposition~\ref{SOSRelaxation} is indeed a relaxation and corresponds to a non-tight description of non-negative polynomials in general.

\section{Data-based representation for nonlinear functions}\label{SecNPModel}

In this section, we introduce our data-based polynomial representation for an unknown nonlinear function based on TP approximation and its derivation from a finite set of noisy samples. This representation will be crucial in the second part of this article to determine dissipativity for dynamical systems from data.

\subsection{Problem setup}\label{SecProblemSetup}

Throughout this section, we study the $k+1$ times continuously differentiable function $f:\mathbb{X}\subset\mathbb{R}^{n_x}\rightarrow\mathbb{R}^{n_y}$ with
\begin{equation*}\label{TrueFunction}
y=f(x)=\begin{bmatrix}f_1(x)&\cdots&f_{n_y}(x)\end{bmatrix}^T
\end{equation*}
and the compact and convex domain
\begin{equation}\label{Domain}
\mathbb{X}=\{x\in\mathbb{R}^{n_x}: w_i(x)\leq 0,\  w_i\in\mathbb{R}[x], i=1,\dots,n_w\}.
\end{equation}
Hence, $\mathbb{X}$ is described by polynomial inequalities. Suppose the function $f$ is unknown but noisy samples
\begin{equation}\label{DataSetFunc}
\{(\tilde{y}_i,\tilde{x}_i)_{i=1,\dots,S}\}
\end{equation}
with $\tilde{y}_i=f(\tilde{x}_i)+\tilde{d}_i$ and unknown disturbance $\tilde{d}_i$ are given. Then the goal of this section is to conclude from the data on an envelope that contains the graph of the function $f$, i.e., $G(f)=\{(x,f(x))\in\mathbb{X}\times\mathbb{R}^{n_y}\}$. \\\indent
To infer on the nonlinear function $f$ from data, we demand some knowledge on the maximal rate of variation of $f$ and the boundedness of the disturbance.

\begin{assu}[Rate of variation]\label{AssBoundDeri}
	Upper bounds $M_{i,\alpha}\geq0,i=1,\dots,n_y,|\alpha|=k+1,$ on the magnitude of each $(k+1)$-th order partial derivative are known
	\begin{align*}
	\bigg|\bigg|\frac{\partial^{k+1}f_i(x)}{\partial x^{\alpha} }\bigg|\bigg|_2\leq M_{i,\alpha},\quad \forall x\in\mathbb{X}.
	\end{align*}
\end{assu} 
\vspace{0.3cm}
\begin{assu}[Pointwise bounded noise]\label{AssNoiseBound}
	The disturbance of the measured data \eqref{DataSetFunc} exhibits a pointwise bound $||\tilde{d}_{i}||_2\leq\epsilon_i(\tilde{y}_i,\tilde{x}_i)$ with $\epsilon_i(\tilde{y}_i,\tilde{x}_i)>0$ for $i=1,\dots,S$.
\end{assu} 

Contrary to, e.g., \cite{Snaizer}, here the structure of $f$ is unknown. Therefore,  we require to bound the rate of variation of the nonlinear function by Assumption~\ref{AssBoundDeri}. Otherwise, an inference from finite samples on a superset of $G(f)$ would be unbounded. Indeed, the nonlinear function could take arbitrarily large values arbitrarily close to a sample without Assumption~\ref{AssBoundDeri}. Hence, a data-based inference on $f(x)$ for an unseen $x$ is not possible. For the same purpose, \cite{Milanese} calls for the Lipschitz constant of $f$ and \cite{MartinDissi} for an upper bound on the degree of the polynomials of the system dynamics. Since the bounds $M_{i,\alpha}$ are usually not known, we refer to Section~\ref{SecUpperBound} for their estimation from data.\\\indent 
The notion of bounded noise in Assumption~\ref{AssNoiseBound} includes noise characterizations with bounded amplitude $\epsilon_i(\tilde{y}_i,\tilde{x}_i)=\epsilon$ and fixed signal-to-noise-ratio $\epsilon_i(\tilde{y}_i,\tilde{x}_i)=\tilde{\epsilon}\,||\tilde{y}_i||_2$. Both characterizations are frequently supposed in data-driven control \cite{vanWaarde} and system analysis \cite{AnneDissi} but also in set-membership identification \cite{Milanese}, adaptive control \cite{AdaptiveC}, and robust model predictive control \cite{RMPC}. If the disturbance $\tilde{d}_i$ is Gaussian distributed, then we refer to the Bayesian treatment in \cite{Gaussian_noise} that is compatible with the dissipativity verification in Section~\ref{SecDissi}.


\subsection{Data-driven envelope by TP approximation}\label{SecModel}

To conclude on an envelope that contains the graph of $f$, we will deduce two ingredients. First, a polynomial approximation error bound between $f$ and its TP is obtained by \textit{Taylor's theorem} \cite{Taylor}. Second, we characterize the set of all polynomials consistent with the polynomial approximation error of TPs for the given data. This set is required as it contains the unknown TP of $f$. \\\indent 
From Taylor's theorem \cite{Taylor}, each element of $f$ can be written as $f_i(x)=T_k(\omega)[f_i(x)]+R_k(\omega)[f_i(x)]$ for some $\omega\in\mathbb{X}$ with the TP
\begin{align*}
T_k(\omega)[f_i(x)]=\sum_{|\alpha|=0}^{k}\frac{1}{\alpha !}\frac{\partial^{|\alpha|}f_i(\omega)}{\partial x^{\alpha}}\left(x-\omega\right)^\alpha={a_i^*}^Tz(x),
\end{align*}
which corresponds to the truncated Taylor series. Here $\alpha!=\alpha_1!\cdots\alpha_{n_x}!$, the vector ${a_i^*}\in\mathbb{R}^{n_z}$ summarizes the unknown coefficients $\frac{\partial^{|\alpha|}f_i(\omega)}{\partial x^{\alpha}}$ of the TP, and $z(x)\in\mathbb{R}[x]^{n_z}$ summarizes all polynomials $\frac{1}{\alpha !}\left(x-\omega\right)^\alpha$ for $|\alpha|=0,\dots,k$. Moreover, for all $x\in\mathbb{R}^{n_x}$ there exists a $\nu\in[0,1]$ such that the Lagrange remainder is given by
\begin{align*}
R_k(\omega)[f_i(x)]{=}\sum_{|\alpha|=k+1}\frac{1}{\alpha !}\frac{\partial^{k+1}f_i(\omega+\nu(x-\omega))}{\partial x^{\alpha}}\left(x-\omega\right)^\alpha.
\end{align*}

Since $\nu$ is deduced from the mean value theorem, its existence is known but usually not its actual value for a specific $x$. Therefore, the next lemma provides two upper bounds on the remainder to circumvent the computation of $\nu$.

\begin{lem}[Bounds on the remainder]\label{LemBound}
	Under Assumption~\ref{AssBoundDeri} and convexity of $\mathbb{X}$, $(R_k(\omega)[f_i(x)])^2\leq R^{\text{abs}}_k(\omega)[f_i(x)]\leq R^{\text{poly}}_k(\omega)[f_i(x)]$ for all $x\in\mathbb{X}$ with
	\begin{align}
	&R^{\text{abs}}_k(\omega)[f_i(x)]= \left(\sum_{|\alpha|=k+1}\frac{M_{i,\alpha}}{\alpha !}||\left(x-\omega\right)^\alpha||_2\right)^2,\label{AbsRemainderBound}\\
	&R^{\text{poly}}_k(\omega)[f_i(x)]=\sum_{|\alpha|=k+1}\kappa_i\frac{M_{i,\alpha}^2}{\alpha !^2}\left(x-\omega\right)^{2\alpha},\label{PolyRemainderBound}	
	\end{align}
	and $\kappa_i\in\mathbb{N}$ is equal to the number of $M_{i,\alpha}\neq0$ for $|\alpha|=k+1$.
\end{lem}
\begin{proof}
	Since $\omega+\nu(x-\omega)\in\mathbb{X}$ for $\nu\in[0,1]$ and for convex set $\mathbb{X}$, $R^{\text{abs}}_k(\omega)[f_i(x)]$ follows by taking the absolute value of each summand of $(R_k(\omega)[f_i(x)])^2$. Subsequently, this implies \eqref{PolyRemainderBound} by multiplying out the square and the fact that $2vw\leq v^2+w^2,\forall v,w\in\mathbb{R}$.	
\end{proof}

By Lemma~\ref{LemBound}, we can conclude on a polynomial characterized sector for $f$ by replacing the non-polynomial remainder through the polynomial bound \eqref{PolyRemainderBound} 
\begin{align}
||f(x)-A^*z(x)||_2^2=&\sum_{i=1}^{n_y}(R_k(\omega)[f_i(x)])^2\notag\\
\leq&\sum_{i=1}^{n_y}R^{\text{abs}}_k(\omega)[f_i(x)]\notag\\
\leq&\sum_{i=1}^{n_y}R^{\text{poly}}_k(\omega)[f_i(x)]\label{PolySecBound}
\end{align}
with $A^*=\begin{bmatrix}a_1^*\hspace{0.2cm} \cdots\hspace{0.2cm}  a_{n_y}^*\end{bmatrix}^T$. As a result, equation \eqref{PolySecBound} describes a polynomial sector bound with the TPs $A^*z(x)$ as center and containing the nonlinear function $f$. This polynomial sector can also been written as the envelope 
\begin{equation}\label{Envelop1}
\left\{(x,y)\in\mathbb{X}\times\mathbb{R}^{n_y}:p_\text{sec}(x,y,A^*)\leq0 \right\}
\end{equation} 
with auxiliary polynomial
\begin{align*}
p_\text{sec}(x,y,A^*)&=||y-A^*z(x)||_2^2-\sum_{i=1}^{n_y}R^{\text{poly}}_k(\omega)[f_i(x)]\\	
&=\star^T
\Phi(\omega)\cdot\begin{bmatrix}\begin{array}{c}\begin{matrix}
I_{n_y} & -I_{n_y} & 0 \\\hline 0 & 0 & 1\end{matrix}\end{array}\end{bmatrix}
\begin{bmatrix}y\\ A^*z(x)\\1	\end{bmatrix},\\
\Phi(\omega)&=\text{diag}\left(I_{n_y}\Big|-\sum\limits_{i=1}^{n_y}R^{\text{poly}}_k(\omega)[f_i(x)]\right).
\end{align*}
By equation \eqref{PolySecBound}, all points $(x,f(x))\in\mathbb{X}\times\mathbb{R}^{n_y}$ satisfy $p_\text{sec}(x,f(x),A^*)\leq0$, and thus envelope \eqref{Envelop1} includes the graph of $f$. Since \eqref{Envelop1} together with the polynomial description of $\mathbb{X}$ achieves a purely polynomial characterized envelope of $G(f)$, equation \eqref{Envelop1} is our first ingredient for a data-based polynomial envelope of $G(f)$. However, the coefficients $A^*$ of the TP at $\omega$ are unknown because $f$ and its derivatives are not available. To this end, we establish a data-based set-membership for TPs. 

\begin{defn}[Set-membership of $A^*$]\label{DefFSS}
	The set of all coefficient matrices $A\in\mathbb{R}^{n_y\times n_z}$ admissible with the measured data \eqref{DataSetFunc} for pointwise bounded noise is given by $\Sigma_{\omega}= \{A:\exists  \tilde{d}_i\in\mathbb{R}^{n_y},i=1,\dots,S, \text{\ satisfying\ } ||\tilde{d}_i||_2\leq\epsilon_i(\tilde{y}_i,\tilde{x}_i)  \text{\ and\ } \tilde{y}_i=Az(\tilde{x}_i)+R_k(\omega)[f(\tilde{x}_i)]+\tilde{d}_i\}$ with $R_k(\omega)[f(x)]=\begin{bmatrix}R_k(\omega)[f_1(x)] & \cdots & R_k(\omega)[f_{n_y}(x)] \end{bmatrix}^T.$
\end{defn}

$\Sigma_{\omega}$ is a set-membership of $A^*$, i.e., $A^*\in\Sigma_{\omega}$, because the samples \eqref{DataSetFunc} satisfy $\tilde{y}_i=A^*z(\tilde{x}_i)+R_k(\omega)[f(\tilde{x}_i)]+\tilde{d}_i$ with $||\tilde{d}_i||_2\leq\epsilon_i(\tilde{y}_i,\tilde{x}_i)$. 
Since the remainder can not intrinsically be evaluated, the following lemma constitutes a superset of $\Sigma_{\omega}$ based on the approximation error bounds for TPs from Lemma~\ref{LemBound}. 

\begin{lem}[Superset of $\Sigma_{\omega}$]\label{LemFSS1}
	The set of coefficient matrices
	\begin{align}\label{DefSetbar}
	\bar{\Sigma}_{\omega}= \{A:
	||\tilde{y}_{i}-Az(\tilde{x}_i)||_2^2\leq q(\tilde{y}_i,\tilde{x}_i) ,i=1,\dots,S\}
	\end{align}
	with $q(\tilde{y}_i,\tilde{x}_i)=\sum_{j=1}^{n_y}R^{\text{abs}}_k(\omega)[f_j(\tilde{x}_i)]+\epsilon_i^2(\tilde{y}_i,\tilde{x}_i)+2\epsilon_i(\tilde{y}_{i},\tilde{x}_{i})\sqrt{ \sum_{j=1}^{n_y}R^{\text{abs}}_k(\omega)[f_j(\tilde{x}_i)]}$ 
	includes $\Sigma_{\omega}$, i.e., $\Sigma_{\omega}\subseteq\bar{\Sigma}_{\omega}$.
\end{lem}
\begin{proof}
	From the definition of $\Sigma_{\omega}$, any $A\in\Sigma_{\omega}$ satisfies
	\begin{align*}
	&\ ||\tilde{y}_{i}-Az(\tilde{x}_i)||_2^2=||R_k(\omega)[f(\tilde{x}_i)]+\tilde{d}_{i}||_2^2\\
	=&\ ||R_k(\omega)[f(\tilde{x}_i)]||_2^2+||\tilde{d}_i||_2^2+2\tilde{d}_i^TR_k(\omega)[f(\tilde{x}_i)]\\
	\leq&\  q(\tilde{y}_i,\tilde{x}_i).
	\end{align*}
	The inequality holds by applying \eqref{AbsRemainderBound} from Lemma~\ref{LemBound}, the noise bound from Assumption~\ref{AssNoiseBound},  and the Cauchy-Schwarz inequality. For that reason, any $A\in\Sigma_{\omega}$ is also included in $\bar{\Sigma}_{\omega}$. This proves the claim.	
\end{proof}

Since $\Sigma_{\omega}$ is a set-membership for $A^*$ and $\Sigma_{\omega}\subseteq\bar{\Sigma}_{\omega}$, $\bar{\Sigma}_{\omega}$ provides a data-based set-membership for $A^*$. Even though we consider in Lemma~\ref{LemFSS1} the non-polynomial but tighter error bound \eqref{AbsRemainderBound} instead of \eqref{PolyRemainderBound}, $\bar{\Sigma}_{\omega}$ is characterized by quadratic constraints on the polynomial coefficients. Indeed, we can reformulate, as in \cite{MartinDissi},
\begin{equation*}
Az(x)=(I_{n_y}\otimes z(x)^T )\text{vec}(A^T)=K(x)\text{vec}(A^T), 
\end{equation*}   
where $\otimes$ denotes the Kronecker product and $\text{vec}$ the vectorization of a matrix by stacking its columns. Then $\bar{\Sigma}_{\omega}$ can be written by quadratic constraints on the elements of $A$ 
\begin{align*}
&\bar{\Sigma}_{\omega}= \{A: \text{for\ } i=1,\dots,S,\\  
&\hspace{0.3cm}\star^T\begin{bmatrix}
\tilde{y}_{i}^T\tilde{y}_{i}-q(\tilde{y}_i,\tilde{x}_i) & -\tilde{y}_{i}^T K(\tilde{x}_i)\\ -K(\tilde{x}_i)^T\tilde{y}_{i} & K(\tilde{x}_i)^TK(\tilde{x}_i)	
\end{bmatrix}\cdot\begin{bmatrix}1\\\text{vec}(A^T)\end{bmatrix}\leq0\}.
\end{align*}
Despite the fact that $\bar{\Sigma}_{\omega}$ would be suitable to determine system properties \cite{MartinDissi}, the vectorization of $A^T$ would lead to computationally demanding SOS optimization problems. Indeed, the size increases with $n_yn_z$, the number of unknown parameters. Since also the number of quadratic constraints increases with the number of samples, we follow the approach from \cite{MartinNLMSOS}. There we combine ellipsoidal outer approximations \cite{Boyd_Elli} and the dualization lemma~\cite{SchererLMI} (Chapter 4.4.1) to construct a more computationally appealing characterization for data-driven set-memberships in case of pointwise bounded noise. 

\begin{lrop}[Superset of $\bar{\Sigma}_{\omega}$]\label{Primal_elli}
	
	If the matrix $\tilde{Z}=\begin{bmatrix}z(\tilde{x}_1) & \cdots & z(\tilde{x}_S)\end{bmatrix}$ has full row rank, then there exist matrices $\Xi_{1\text{p}}\succ0$ and $\Xi_{2\text{p}}$ and scalars $\eta_1,\dots,\eta_S\geq0$ solving the LMI
	\begin{align}\label{Elli_LMI}
	\begin{bmatrix}\begin{array}{c|c}
	\begin{matrix}\Xi_{1\text{p}} & \Xi_{2\text{p}}\\\Xi_{2\text{p}}^T & -I_{n_y}	\end{matrix}& \begin{matrix}0\\\Xi_{2\text{p}}^T \end{matrix} \\\hline 
	\begin{matrix}0 &\Xi_{2\text{p}}\end{matrix} & -\Xi_{1\text{p}} \end{array}	\end{bmatrix}-\sum_{i=1}^{S}\eta_i\begin{bmatrix}\begin{array}{c|c}\Xi_i &  0\\\hline   0 & 0\end{array}\end{bmatrix}
	\prec0,
	\end{align}
	for the data-based matrices
	\begin{equation*}
	\Xi_i=\begin{bmatrix}z(\tilde{x}_i)z(\tilde{x}_i)^T & -z(\tilde{x}_i)\tilde{y}_i^T\\ -\tilde{y}_iz(\tilde{x}_i)^T & \tilde{y}_i\tilde{y}_i^T-q(\tilde{y}_i,\tilde{x}_i)I_{n_y}\end{bmatrix}.
	\end{equation*}
	Moreover, $\bar{\Sigma}_{\omega}$ is a subset of 
	\begin{align}\label{Sigma_A}
	\tilde{\Sigma}_{\omega}=\left\{A:\begin{bmatrix}I_{n_z}\\A\end{bmatrix}^T\varDelta_*\begin{bmatrix}I_{n_z}\\A\end{bmatrix}\preceq0\right\}
	\end{align}
	with $\varDelta_*=\begin{bmatrix}-\varDelta_{1\text{p}} & \varDelta_{2\text{p}}\\ \varDelta_{2\text{p}}^T & -\varDelta_{3\text{p}}\end{bmatrix}$, $\varDelta_{1\text{p}}\in\mathbb{R}^{n_z\times n_z},\varDelta_{2\text{p}}\in\mathbb{R}^{n_z\times n_y},\varDelta_{3\text{p}}\in\mathbb{R}^{n_y\times n_y}$, $\begin{bmatrix}\varDelta_{1\text{p}} & \varDelta_{2\text{p}}\\ \varDelta_{2\text{p}}^T & \varDelta_{3\text{p}}\end{bmatrix}=\Xi_\text{p}^{-1}$, and $\Xi_\text{p}=\begin{bmatrix}\Xi_{1\text{p}} & \Xi_{2\text{p}}\\ \Xi_{2\text{p}}^T & \Xi_{2\text{p}}^T\Xi_{1\text{p}}^{-1}\Xi_{2\text{p}}-I_{n_y}\end{bmatrix}$.
\end{lrop}

\begin{proof}
	We first prove that LMI \eqref{Elli_LMI} has a solution if $\tilde{Z}$ has full row rank by adapting the proof of Lemma 2 in \cite{Ellipsoid_DePersis2}. To this end, let $\eta>0$ be a to-be-optimized scalar and set $\eta_i=\eta$ for all $i=1,\dots,S$, $\Xi_{1\text{p}}=\frac{\eta}{2} \tilde{Z}\tilde{Z}^T$, and $\Xi_{2\text{p}}=-\eta \tilde{Z}\tilde{Y}^T$ with $\tilde{Y}=\begin{bmatrix}\tilde{y}_1 & \cdots & \tilde{y}_S\end{bmatrix}$. For $\tilde{Z}\tilde{Z}^T=\sum_{i=1}^{S} z(\tilde{x}_i)z(\tilde{x}_i)^T$ and $\tilde{Z}\tilde{Y}=\sum_{i=1}^{S} z(\tilde{x}_i)\tilde{y}_i^T$ together with the choice of $\Xi_{1\text{p}}$ and $\Xi_{2\text{p}}$, LMI $\eqref{Elli_LMI}$ yields  
	\begin{align*}
	\begin{bmatrix}-\frac{\eta}{2} \tilde{Z}\tilde{Z}^T &0&0\\0&-I_{n_y}+\eta\sum_{i=1}^{S}(\tilde{y}_i\tilde{y}_i^T-qI_{n_y}) & \Xi_{2\text{p}}^T\\ 0 &\Xi_{2\text{p}} & -\frac{\eta}{2} \tilde{Z}\tilde{Z}^T \end{bmatrix}\prec0.
	\end{align*}
	Since $\tilde{Z}$ has full row rank, $\tilde{Z}\tilde{Z}^T\succ0$, and hence the LMI is satisfied if $I_{n_y}-\eta\sum_{i=1}^{S}(\tilde{y}_i\tilde{y}_i^T-q(\tilde{y}_i,\tilde{x}_i)I_{n_y})-2\eta(\tilde{Z}\tilde{Y}^T)^T(\tilde{Z}\tilde{Z}^T)^{-1}\tilde{Z}\tilde{Y}^T\succ0$ by the Schur complement. This is always feasible by choosing $\eta>0$ sufficiently small.\\\indent
	Next we show the second statement $\bar{\Sigma}_{\omega}\subseteq\tilde{\Sigma}_{\omega}$. By its definition \eqref{DefSetbar}, $\bar{\Sigma}_{\omega}$ can also be characterized by the quadratic constraints	
	\begin{align*}
	\star^T
	\begin{bmatrix}
	\tilde{y}_i^T\tilde{y}_i-q(\tilde{y}_i,\tilde{x}_i) & -\tilde{y}_i^T\\ -\tilde{y}_i & I_{n_y}\end{bmatrix}\cdot	\begin{bmatrix}1\\ Az(\tilde{x}_i)	\end{bmatrix}\leq0,\ i=1,\dots,S.
	\end{align*}
	Since the inverse of the inner matrix is given by $\frac{-1}{q(\tilde{y}_i,\tilde{x}_i)}\begin{bmatrix}1 & \tilde{y}_i^T\\ \tilde{y}_i & \tilde{y}_i\tilde{y}_i^T-q(\tilde{y}_i,\tilde{x}_i)I_{n_y}\end{bmatrix}$
	and $q(\tilde{y}_i,\tilde{x}_i)>0$ by Assumption~\ref{AssNoiseBound}, the dualization lemma~\cite{SchererLMI} (Chapter 4.4.1) yields the equivalent description
	\begin{align}\label{Sigma_dual}
	\bar{\Sigma}_{\omega}=\left\{A:\begin{bmatrix}A^T\\I_{n_y}\end{bmatrix}^T\Xi_i\begin{bmatrix}A^T\\I_{n_y}\end{bmatrix}\prec0,i=1,\dots,S\right\}.
	\end{align}
	Observe that \eqref{Sigma_dual} is given by the intersection of matrix ellipsoids as in \cite{MartinNLMSOS} and \cite{Ellipsoid_DePersis}. Hence, an ellipsoidal outer approximation of the intersection of all matrix ellipsoids of each sample is calculated by \eqref{Elli_LMI} following the result from \cite{Boyd_Elli}. This yields the superset $\left\{A:\begin{bmatrix}A^T\\I_{n_y}\end{bmatrix}^T\Xi_\text{p}\begin{bmatrix}A^T\\I_{n_y}\end{bmatrix}\prec0\right\}\supseteq\bar{\Sigma}_{\omega}$ which is equivalent to $\tilde{\Sigma}_{\omega}$ by the dualization lemma.\\\indent
	It remains to prove the invertibility of $\Xi_\text{p}$. To this end, suppose $\Xi_\text{p}$ has not full rank. Then there exists a vector $r=\begin{bmatrix}r_1^T& r_2^T\end{bmatrix}^T\neq0$ such that
	\begin{align*}
	\Xi_p\begin{bmatrix}r_1\\ r_2\end{bmatrix}=\begin{bmatrix}\Xi_{1\text{p}}r_1+\Xi_{2\text{p}}r_2\\ \Xi_{2\text{p}}^T(r_1+\Xi_{1\text{p}}^{-1}\Xi_{2\text{p}}r_2)-r_2\end{bmatrix}=0.
	\end{align*}
	Combining both equations yields $\Xi_{2\text{p}}^T(r_1-\Xi_{1\text{p}}^{-1}\Xi_{1\text{p}}r_1)-r_2=-r_2=0$. Together with $\Xi_{1\text{p}}\succ0$, $\Xi_{1\text{p}}r_1+\Xi_{2\text{p}}r_2=0$ implies $r_1=0$, and hence $r=0$. Due to the contradiction with $r\neq0$, $\Xi_\text{p}$ has full rank, and therefore is invertible.
\end{proof}

The LMI feasibility problem \eqref{Elli_LMI} can be extended by maximizing over $\gamma>0$ with $\Xi_{1\text{p}}\succeq\gamma I_{n_z}$ to retrieve an ellipse with minimal diameter. Moreover, equation \eqref{Sigma_dual} is a dual description as it contains $A^T$ instead of $A$. Furthermore, $\Xi_i$ are not invertible as their left upper block have rank one. Thus, the dualization lemma can not be employed directly to recover a primal characterization. Instead, we first compute in Proposition~\ref{Primal_elli} an ellipsoidal outer approximation of the intersection of all matrix ellipsoids of each sample and then apply the dualization lemma.\\\indent
We assess the assumptions in Proposition~\ref{Primal_elli} to be not restrictive as the row rank of the data-dependent matrix $\tilde{Z}$ can be imposed by increasing the number of columns with additional samples. The rank criterion can easily be checked from data. Also note that the rank condition can be explained geometrically as in the following remark. 

\begin{rmk}[Geometric explanation of rank condition of $\tilde{Z}$]\label{RmkGeo}
	Observe that $\bar{\Sigma}_{\omega}$ is characterized by the intersection of the matrix ellipsoids from \eqref{Sigma_dual} which can also be written as
	\begin{equation}\label{Ellip}
	\star^Tz(\tilde{x}_i)z(\tilde{x}_i)^T\cdot\left(A^T-\frac{z(\tilde{x}_i)\tilde{y}_i^T}{z(\tilde{x}_i)^Tz(\tilde{x}_i)}\right)\preceq q(\tilde{y}_i,\tilde{x}_i)I_{n_y}.
	\end{equation}
	Since $z(\tilde{x}_i)z(\tilde{x}_i)^T$ has rank one, equation \eqref{Ellip} is a quadratic constraint on $A^T$ that does not describe a closed matrix ellipse. Instead, \eqref{Ellip} corresponds to the unbounded set with boundaries $A_B^T=\frac{z(\tilde{x}_i)\tilde{y}_i^T}{z(\tilde{x}_i)^Tz(\tilde{x}_i)}+B\frac{\sqrt{q(\tilde{y}_i,\tilde{x}_i)}}{z(\tilde{x}_i)^Tz(\tilde{x}_i)}z(\tilde{x}_i)\frac{\ell^T}{||\ell||_2}+\sum_{i=1}^{n_z-1}\sum_{j=1}^{n_y}c_{ij}v_iw_j^T$ where $B\in\{-1,1\}$, $\ell\in\mathbb{R}^{n_y}$ and $c_{ij}\in\mathbb{R}$ are arbitrary, $v_1,\dots,v_{n_z-1}\in\mathbb{R}^{n_z}$ span an orthogonal basis to $z(\tilde{x}_i)$, i.e., $z(\tilde{x}_i)^Tv_i=0,$ and $w_1,\dots,w_{n_y}\in\mathbb{R}^{n_y}$ span a basis of $\mathbb{R}^{n_y}$. Indeed, for $A^T=A_B^T$, equation \eqref{Ellip} yields $\frac{\ell\ell^T}{||\ell||_2^2}\preceq I_{n_y}$ which is satisfied as the non-zero eigenvalue of $\frac{\ell\ell^T}{||\ell||_2^2}$ is equal to $\frac{\ell^T\ell}{||\ell||_2^2}=1$. Moreover, a perturbation $A_B^T+B\epsilon z(\tilde{x}_i)\ell^T$ for arbitrary small $\epsilon>0$ violates \eqref{Ellip}. By the boundaries $A_B^T$, if a matrix $\tilde{A}^T$ satisfies \eqref{Ellip} then $\tilde{A}^T+\sum_{i=1}^{n_z-1}\sum_{j=1}^{n_y}\tilde{c}_{ij}v_iw_j^T$ for arbitrary $\tilde{c}_{ij}\in\mathbb{R}$ as well. Together with the fact that $v_1,\dots,v_{n_z-1}\in\mathbb{R}^{n_z}$ span an orthogonal basis to $z(\tilde{x}_i)$, the intersection of the unbounded sets described by \eqref{Ellip} for $i=1,\dots,S$ is bounded if $z(\tilde{x}_1),\dots,z(\tilde{x}_S)$ span the whole $\mathbb{R}^{n_z}$. This is equivalent to $\tilde{Z}$ having full row rank. Thus, there exists an ellipsoidal outer approximation in \eqref{Elli_LMI} with $\Xi_{1p}\succ0$.
\end{rmk}

Finally, we can combine the two ingredients, the envelope \eqref{Envelop1} and the set-membership \eqref{Sigma_A} of the TP of $f$ at $\omega$, to constitute a data-based envelope of $G(f)$.
\begin{thm}[Data-based envelope of $G(f)$]\label{ThmDataEnvelope}
	Under Assumption~\ref{AssBoundDeri} and \ref{AssNoiseBound} and $\tilde{Z}=\begin{bmatrix}z(\tilde{x}_1) & \cdots & z(\tilde{x}_S)\end{bmatrix}$ having full row rank, the graph of a $k+1$ times continuously differentiable function $f:\mathbb{X}\subset\mathbb{R}^{n_x}\rightarrow\mathbb{R}^{n_y}$ is a subset of
	\begin{align}\label{EnvelopDB}
	E(f) = \bigcup_{A\in\tilde{\Sigma}_{\omega}}\left\{(x,y)\in\mathbb{X}\times\mathbb{R}^{n_y}:p_\text{sec}(x,y,A)\leq0 \right\}.
	\end{align}
\end{thm}
\begin{proof}
	Since $\tilde{\Sigma}_{\omega}$ contains the coefficient matrix $A^*$ of the TP of $f$ at $\omega$ and the envelope \eqref{Envelop1} includes $G(f)$ for $A^*$, the statement $G(f)\subseteq E(f)$ holds true.	
\end{proof}

We emphasize that the envelope $E(f)$ is a data-based polynomial representation of the nonlinear function $f$ that does not require the knowledge of a function basis containing $f$. Indeed, $p_\text{sec}$ and $w_1,\dots,w_{n_w}$ in \eqref{Domain} are polynomials and $\tilde{\Sigma}_{\omega}$ is deduced from data~\eqref{DataSetFunc} by means of Proposition~\ref{Primal_elli}. Moreover, $E(f)$ includes for all $M_{i,\alpha}=0$ the data-driven models from \cite{AnneDissi,WaardeDissi,MartinDissi} where a data-based set-membership for linear and polynomial systems, respectively, are introduced. Contrary to a linear or polynomial dynamics which depend linearly on unknown parameters, we consider here dynamics with unknown structure. Thereby, the presented method provides a new viewpoint to deal with more complex unknown dynamics: we establish a set-membership for an unknown approximation of the dynamics by means of a characterization of its approximation error.\\\indent
Classical polynomial interpolation of nonlinear functions supposes that the nonlinear function and its derivatives can be exactly evaluated at specified points \cite{Taylor} and \cite{BernsteinPoly}. Contrary, here we only have access to noisy samples of $f$ at unspecified points. Due to this additional flexibility, $E(f)$ is suitable for the data-based analysis of dynamical systems. There we typically have only access to noisy measurements of trajectories, and hence we can not specify exactly where the system dynamics is evaluated.\\\indent
The envelope $E(f)$ can also be defined with $\Sigma_{\omega}$ or $\bar{\Sigma}_{\omega}$. However, the former requires the evaluation of the remainder $R_k(\omega)[f(x)]$ and the latter renders $E(f)$ complex for a large amount of samples.\\\indent
In general, the conservatism of $E(f)$ is obscure as the size of $\tilde{\Sigma}_\omega$ is even for LTI systems unclear. To this end, \cite{MartinNL} depicts $E(f)$ for two numerical examples. There we observe that the data-driven inferences on system properties using tight bounds for Assumption~\ref{AssBoundDeri} and \ref{AssNoiseBound} are almost equal to the inference from the true TP. Therefore, $\tilde{\Sigma}_\omega$ is expected to be small under these conditions. Moreover, the conservatism of $\tilde{\Sigma}_\omega$ is reduced if the available samples are close to $\omega$. Indeed, the distance between the boundaries $A^T_{1}$ and $A^T_{-1}$ from Remark~\ref{RmkGeo} can be measured by the Frobenius norm 
$\Big|\Big|2\frac{\sqrt{q(\tilde{y}_i,\tilde{x}_i)}}{z(\tilde{x}_i)^Tz(\tilde{x}_i)}z(\tilde{x}_i)\frac{\ell^T}{||\ell||_2}\Big|\Big|_\text{Fr}^2=\text{trace}\left(4\frac{{q(\tilde{y}_i,\tilde{x}_i)}}{z(\tilde{x}_i)^Tz(\tilde{x}_i)}\frac{\ell\ell^T}{||\ell||_2^2}\right)=4\frac{{q(\tilde{y}_i,\tilde{x}_i)}}{z(\tilde{x}_i)^Tz(\tilde{x}_i)}$. Since $q$ and $z^Tz$ contain $||x-\omega||_2$ of order at most $k+1$ and $k$, respectively, the distance increases with $||x-\omega||_2$. Thereby, 
data close to $\omega$ provide tighter boundaries, and hence generally more information on the set-membership $\tilde{\Sigma}_\omega$. For that reason, choosing $\omega$ where more samples are available, can yield a less conservative envelope.\\\indent
Whereas the conservatism of $E(f)$ by $\tilde{\Sigma}_\omega$ depends strongly on the given data, the conservatism by the polynomial sector bound by TP approximation is known by \eqref{PolySecBound}. Therefore, the smallest possible $E(f)$ corresponds to \eqref{PolySecBound}. Moreover, if the smoothness of $f$ allows to consider the order of the TP approximation $k$ as degree of freedom, then considering the polynomial sector \eqref{PolySecBound} for $\mathbb{X}$ can provide insights on how to chose $k$. In fact, 
while the sector for larger $k$ is smaller close to $\omega$, the sector is wider for larger values of $||x-\omega||_2$. Thus, a larger $k$ does not necessarily result in a more accurate inference on a system property. Besides the size of $\mathbb{X}$, the choice of $k$ also depends on the available data because noisy data typically provide only conservative inferences on coefficients of higher order monomials. Concluding, we recommend to stepwise increase $k$ for data-driven inferences on system properties as it is not a priori clear whether larger $k$ yields less conservative inferences. Furthermore, if the maximal allowed $k$ is not known, we refer to Remark~\ref{RmkChoiceofk} for the derivation of a meaningful selection of $k$.

\subsection{Improving the accuracy of envelope E(f)}\label{Sec_Improv}

Based on the previous discussion, the accuracy of $E(f)$ can be significantly improved by combining multiple envelopes with TPs around distinct points $\omega_{1},\dots,\omega_{n_\omega}$
\begin{align}\label{InterTSE}
E'(f)= \bigcap_{i=1}^{n_\omega}\bigcup_{A\in\tilde{\Sigma}_{\omega_i}}\left\{(x,y)\in\mathbb{X}\times\mathbb{R}^{n_y}:p_\text{sec}(x,y,A)\leq0 \right\}.
\end{align}
Hence, the idea of $E'(f)$ is to represent a nonlinear function by multiple polynomial approximations around $\omega_{i}$ rather than finding one polynomial approximation. A piecewise polynomial representation is similar to a spline approximation which exhibits a better approximation than a single polynomial \cite{Poly_Spline}. Hence, a significant improvement of estimating system properties by $E'(f)$ compared to $E(f)$ is to be expected at the cost of higher computational complexity. Whereas a uniform grid over $\mathbb{X}$ for the selection of $\omega_i$ leads to a minimal approximation error by polynomial sector \eqref{PolySecBound}, an iteratively obtained selection of $\omega_i$ for optimized inferences on system properties will be presented in Remark~\ref{RmkOmega}.  \\\indent
$E'(f)$ generalizes the set-membership approach in \cite{Milanese} and Kinky inference \cite{KinkyInf} where an envelope for nonlinear functions is derived from Lipschitz approximations. Indeed, the polynomial sector bound \eqref{PolySecBound} boils down for $k=0$ to the usual Lipschitz continuity condition as $A^*z(x)=f(\omega)$. Moreover, Assumption~\ref{AssBoundDeri} implies a Lipschitz constant of $f$. \\\indent 
A purely data-driven approach for nonlinear functions might be excessively conservative. To additionally reduce the conservatism and complexity of the envelopes $E(f)$ and $E'(f)$, we can leverage prior knowledge on the structure of $f$. For instance, if $f_1$ is only a polynomial with degree $d_1$, then the TP can be developed until order $d_1$ and $M_{1,\alpha}$ for $|\alpha|=d_1+1$ can be set to zero. Moreover, to show stability of a known equilibrium of a dynamical system, we can enforce $f(0)=0$ by omitting the $0$-th order of the TP at $\omega=0$. In the sequel, we present how the separate consideration of data-driven envelopes for each $f_i$ allows to exploit the knowledge that the TPs of $f_i$ might be given by $a_i^*z_i(x),i=1,\dots,n_x$. This can reduce the number of unknown coefficients compared to Section~\ref{SecModel} with $z_1=\dots=z_{n_x}=z$ because $z_i$ might not depend on all $x_1,\dots,x_{n_x}$.\\\indent
For the sake of illustration, let a nonlinear function $f$ be given by 
\begin{equation*}
f(x)=\begin{bmatrix}f_1(x)\\f_2(x)\end{bmatrix}=\begin{bmatrix}x_1+x_2+g(x_2)\\\mu^* x_1^2+3g(x_2)\end{bmatrix},
\end{equation*}
where $\mu^*\in\mathbb{R}$ and the nonlinearity $g(x_2)$ are unknown. We want to exploit the knowledge that some coefficients of $f_1$ and $f_2$ are known, that both functions contain $g$, and that $g$ only depends on $x_2$. By Taylor's theorem, let $g(x_2)=T_k(\omega)[g(x_2)]+R_k(\omega)[g(x_2)]$, where $T_k(\omega)[g(x_2)]={b^*}^Tz(x_2)$ with unknown vector $b^*\in\mathbb{R}^{n_z}$, and bounds on $R_k(\omega)[g(x_2)]$ be known as in Lemma~\ref{LemBound}. Then equation \eqref{PolySecBound} directly amounts to the polynomial sector bounds
\begin{align}
\left(f_1(x)-x_1-x_2-\begin{bmatrix}\mu^* & {b^*}^T\end{bmatrix}\begin{bmatrix}0 \\ z(x_2)\end{bmatrix}\right)^2{\leq} R_k^\text{poly}(\omega)[g(x_2)]\label{PolySecf1},\\
\left(f_2(x)-\begin{bmatrix}\mu^* & {b^*}^T\end{bmatrix}\begin{bmatrix}x_1^2 \\ 3z(x_2)\end{bmatrix}\right)^2{\leq}9 R_k^\text{poly}(\omega)[g(x_2)].\label{PolySecf2}
\end{align}
At the same time, let samples $\{(\tilde{y}_i,\tilde{x}_i)_{i=1,\dots,S}\}$ with $\tilde{x}_i=\begin{bmatrix}\tilde{x}_{i1} & \tilde{x}_{i2}\end{bmatrix}^T\in\mathbb{R}^2$, $\tilde{y}_i=\begin{bmatrix}\tilde{y}_{i1} & \tilde{y}_{i2}\end{bmatrix}^T\in\mathbb{R}^2$, and $\tilde{y}_{ij}=f_j(\tilde{x}_i)+\tilde{d}_{ij}$ with $|\tilde{d}_{ij}|\leq\epsilon_{ij}(\tilde{y}_{i},\tilde{x}_{i}),j=1,2,$ be accessible. Then \eqref{AbsRemainderBound} yields for all $i=1,\dots,S$
\begin{align*}
&\left(\tilde{y}_{i1}-\tilde{x}_{i1}-\tilde{x}_{i2}-\begin{bmatrix}0 & z(\tilde{x}_{i2})^T\end{bmatrix}\begin{bmatrix}\mu^* \\ {b^*}\end{bmatrix}\right)^2\\
=&\, (R_k(\omega)[g(\tilde{x}_{i2})]+\tilde{d}_{i1})^2\\
\leq& \left(\sqrt{R^{\text{abs}}_k(\omega)[g(\tilde{x}_{i2})]}+\epsilon_{i1}(\tilde{y}_i,\tilde{x}_i)\right)^2
\end{align*}
and $\left(\tilde{y}_{i2}-\begin{bmatrix}\tilde{x}_{i2}^2 & 3z(\tilde{x}_{i2})^T\end{bmatrix}\begin{bmatrix}\mu^* \\ {b^*}\end{bmatrix}\right)^2	\leq \left(3\sqrt{R^{\text{abs}}_k(\omega)[g(\tilde{x}_{i2})]}+\epsilon_{i2}(\tilde{y}_i,\tilde{x}_i)\right)^2$.
A reformulation directly leads to $2S$ quadratic constraints in the dual space
\begin{align*}
\begin{bmatrix}\begin{bmatrix}\mu^* \\ {b^*}\end{bmatrix}\\ 1	\end{bmatrix}^T\Xi_{ij}	
\begin{bmatrix}\begin{bmatrix}\mu^* \\ {b^*}\end{bmatrix}\\ 1	\end{bmatrix}\leq0,i=1,\dots,S, j=1,2.
\end{align*}
Hence, we receive by means of ellipsoidal outer approximation \eqref{Elli_LMI} and the dualization lemma a set-membership for $\begin{bmatrix}\mu^* & {b^*}\end{bmatrix}$
\begin{align}\label{SetMemfj}
\tilde{\Sigma}_{\omega}=\left\{\begin{bmatrix}\mu \\ {b}\end{bmatrix}:\begin{bmatrix}I_{n_z}\\\begin{bmatrix}\mu & {b}\end{bmatrix}	\end{bmatrix}^T
\varDelta
\begin{bmatrix}I_{n_z}\\\begin{bmatrix}\mu & {b}\end{bmatrix}	\end{bmatrix}\preceq0\right\}.
\end{align}
Lastly, combining the set-membership \eqref{SetMemfj} and the sector bounds \eqref{PolySecf1} and \eqref{PolySecf2} attains a data-based envelope $E(f_1)$ and $E(f_2)$, respectively, analogously to \eqref{EnvelopDB}. \\\indent
The separated envelopes $E(f_1)$ and $E(f_2)$ reduce the conservatism of the derived envelope for $f$ compared to the envelope \eqref{EnvelopDB} with a cumulative polynomial sector bound and set-membership for the TPs of $f$. Since $f_1$ and $f_2$ are only nonlinear with respect to $x_2$, $\omega_i$ for $E'(f)$ should only be distributed in the $x_2$-direction rather than $x_1$.

\subsection{Validation of $\mathit{M_{i,\alpha}}$ from data}\label{SecUpperBound}

The envelope $E(f)$ of Subsection~\ref{SecModel} is based on the knowledge of $M_{i,\alpha}$ from Assumption~\ref{AssBoundDeri}. However, these are usually not known. Since finite differences are sensitive regarding to the noise of samples, we suggest a validation procedure to estimate $M_{i,\alpha}$ from data in a sound way. For that purpose, we provide a guaranteed risk that an unseen sample violates the envelope $E(f)$.\\\indent
Since we can not verify from samples whether some bounds $M_{i,\alpha}$ satisfy Assumption~\ref{AssBoundDeri}, we can not calculate bounds $M_{i,\alpha}$ from data. However, we can validate whether the data-based envelope $E(f)$ from data \eqref{DataSetFunc} for some $M_{i,\alpha}$ explains a set of additional validation data
\begin{equation}\label{DataSetVal}
\{(\tilde{y}_i,\tilde{x}_i)_{i=S+1,\dots,S+V}\}
\end{equation}
with $\tilde{y}_i=f(\tilde{x}_i)+\tilde{d}_i$ and $||\tilde{d}_i||_2\leq\epsilon,i=S+1,\dots,S+V$. Therefore, we suggest in the following a validation procedure to gather suitable bounds $M_{i,\alpha}$.\\\indent 
To this end, observe that if $M_{i,\alpha}$ are chosen too small then the obtained envelope $E(f)$ from data \eqref{DataSetFunc} might be not admissible with the data \eqref{DataSetVal} or the set of feasible polynomials $\tilde{\Sigma}_{\omega}$ might even be empty, i.e., LMI \eqref{Elli_LMI} is infeasible. Note that even if the matrix $\tilde{Z}$ has full row rank, the LMI \eqref{Elli_LMI} can be infeasible as $M_{i,\alpha}$ is supposed to satisfy Assumption~\ref{AssBoundDeri} in Proposition~\ref{Primal_elli}. On the other hand, if function $f$ is $k+1$ times continuously differentiable, then the existence of bounds $M_{i,\alpha}$ on the compact domain $\mathbb{X}$ is guaranteed. Hence, there exist $M_{i,\alpha}$ that render $E(f)$ to be consistent with the data \eqref{DataSetFunc} and \eqref{DataSetVal}. Due to these observations, we consider a set-membership validation procedure \cite{Milanese} (Definition 2) to find suitable $M_{i,\alpha}$. Our procedure additionally achieves a guaranteed bound on the risk that an unseen sample is not contained in $E(f)$.
\begin{algo}[Set-membership validation procedure]\label{SVP}\indent
	\begin{itemize}
		\item[$0.)$] Set $\underline{M}:=0$ and choose stopping criterion $\delta>0$, a confidence level $c\in(0,1)$, and maximal risk $\bar{\mu}\in(0,1)$.
		Find $\bar{M}$ such that $\tilde{\Sigma}_{\omega}\neq\emptyset$ from data \eqref{DataSetFunc} by solving \eqref{Elli_LMI} for $M_{i,\alpha}=\bar{M},i=1,\dots,n_y,|\alpha|=k+1,$ and $\tilde{\mu}+\epsilon_c\leq\bar{\mu}$ with $\epsilon_c=\sqrt{\frac{\text{ln}(2/c)}{2V}}$,
		\begin{align}
		\tilde{\mu}&=\frac{1}{V}\sum_{i=S+1}^{S+V}I(\tilde{x}_i,\tilde{y}_i),\label{EmpRisk}\\
		I(\tilde{x}_i,\tilde{y}_i)&=\left\{ \begin{array}{cl} 
		0&	 \text{if } \{\tilde{x}_i\}\times\mathcal{B}_{\epsilon}(\tilde{y}_i)\subseteq E(f) \\
		1& \text{otherwise, } \end{array} \right..\label{IndicatorFcn}		
		\end{align}
		$\mathcal{B}_r(\xi)$ denotes the ball around $\xi$ with radius $r>0$.
		\item[$1.)$] Set $M_{i,\alpha}:=\frac{\bar{M}+\underline{M}}{2},i=1,\dots,n_y,|\alpha|=k+1$. Compute $\tilde{\Sigma}_{\omega}$ from data \eqref{DataSetFunc} by solving \eqref{Elli_LMI} and $\tilde{\mu}$ from \eqref{EmpRisk}.
		\item[$2.)$] If $\tilde{\Sigma}_{\omega}\neq\emptyset$ and $\tilde{\mu}+\epsilon_c\leq\bar{\mu}$, then set $\bar{M}:=\frac{\bar{M}+\underline{M}}{2}$. Otherwise, set $\underline{M}:=\frac{\bar{M}+\underline{M}}{2}$.
		\item[$3.)$] If $\bar{M}-\underline{M}<\delta$, then set $M_{i,\alpha}=\bar{M},i=1,\dots,n_y,|\alpha|=k+1,$ and stop. Otherwise go to step 1.).
	\end{itemize}
\end{algo}

In terms of the previous comments, it is clear that there always exists an $\bar{M}$ in the initialization step $0.)$. Moreover, the algorithm is based on a bisection procedure to find the smallest $M_{i,\alpha}=\bar{M},i=1,\dots,n_y,|\alpha|=k+1,$ such that the envelope $E(f)$ explain the data \eqref{DataSetFunc} and partially \eqref{DataSetVal}. However, the estimated $M_{i,\alpha}$ nonetheless might violate Assumption~\ref{AssBoundDeri} as $\tilde{\Sigma}_{\omega}$ is not a tight description of ${\Sigma}_{\omega}$ and the noise characterization of Assumption~\ref{AssNoiseBound} might also be conservative. Hence, for example, if the envelope is admissible with the data for $\bar{M}=0$ then $f$ might be polynomial, the data might be too sparse, or the conservatism of $\tilde{\Sigma}_{\omega}$ and the noise bound from Assumption~\ref{AssNoiseBound} might cover the non-polynomial nonlinearities. Due to the inclusion of the noise in the computation of $\tilde{\Sigma}_{\omega}$, it is expected that the set-membership validation procedure is less sensitive regarding noise than estimating the high order partial derivatives of $f$ by finite differences. \\\indent 
To check efficiently whether the set $\{\tilde{x}_i\}\times\mathcal{B}_{\epsilon}(\tilde{y}_i)$ is included in $E(f)$, consider the set of all outputs $y$ with $(\tilde{x}_i,y)\in E(f)$ for a validation point $\tilde{x}_i$. By the definition of $E(f)$ in \eqref{EnvelopDB} and definition of $p_\text{sec}$ in \eqref{Envelop1}, this set of outputs is given by the union of balls $\mathcal{B}_{\sum_{j=1}^{n_y}R^{\text{poly}}_k(\omega)[f_j(\tilde{x}_i)]}(Az(\tilde{x}_i))$. The center $Az(\tilde{x}_i)$ corresponds to the ellipse
\begin{align*}
&\star^T(-\Delta_{3\text{p}})\cdot(Az(\tilde{x}_i)-\Delta_{3\text{p}}^{-1}\Delta_{2\text{p}}z(\tilde{x}_i))\\
\leq& z(\tilde{x}_i)^T(\Delta_{1\text{p}}+\Delta_{2\text{p}}^T\Delta_{3\text{p}}^{-1}\Delta_{2\text{p}})z(\tilde{x}_i),
\end{align*}
which follows by left- and right-multiplication of \eqref{Sigma_A} by $z(\tilde{x}_i)$. 
Note that this ellipsoid contains  $\mathcal{B}_{{z(\tilde{x}_i)^T(\Delta_{1\text{p}}+\Delta_{2\text{p}}^T\Delta_{3\text{p}}^{-1}\Delta_{2\text{p}})z(\tilde{x}_i)}/\lambda_{\text{max}}(-\Delta_{3\text{p}})}(\Delta_{3\text{p}}^{-1}\Delta_{2\text{p}}z(\tilde{x}_i))$ since $\Delta_{3\text{p}}\prec0$ follows from  the dualization lemma in Proposition~\ref{Primal_elli}. Hence, $I(\tilde{x}_i,\tilde{y}_i)=0$ if $\mathcal{B}_{\epsilon}(\tilde{y}_i)\subseteq\mathcal{B}_{\tilde{r}_i}(\Delta_{3\text{p}}^{-1}\Delta_{2\text{p}}z(\tilde{x}_i))$ with radius $\tilde{r}_i=\sum_{i=j}^{n_y}R^{\text{poly}}_k(\omega)[f_j(\tilde{x}_i)]+{z(\tilde{x}_i)^T(\Delta_{1\text{p}}+\Delta_{2\text{p}}^T\Delta_{3\text{p}}^{-1}\Delta_{2\text{p}})z(\tilde{x}_i)}/\lambda_{\text{max}}(-\Delta_{3\text{p}})$ which is equivalent to
\begin{equation}\label{CheckInclusion}
||\tilde{y}_i-\Delta_{3\text{p}}^{-1}\Delta_{2\text{p}}z(\tilde{x}_i)||_2+\epsilon\leq \tilde{r}_i.
\end{equation}
Summarizing, Algorithm~\ref{SVP} requires in each iteration to solve the LMI~\eqref{Elli_LMI} and to check \eqref{CheckInclusion} for each validation data point from \eqref{DataSetVal}. While the complexity of both conditions scales linearly in the number of training and validation data, respectively, the bisection converges within few iterations in our examples. Hence, the computation time of Algorithm~\ref{SVP} is typically smaller than the time to solve the SOS optimization problem for verifying dissipativity from Section~\ref{SecDissi}. \\\indent
In case the data \eqref{DataSetFunc} and \eqref{DataSetVal} cover densely the whole domain $\mathbb{X}$, it is expected that the derived estimations of $M_{i,\alpha}$ are reasonable and $G(f)$ is included in $E(f)$. In Algorithm~\ref{SVP}, the accuracy of the envelope $E(f)$ for the graph of $f$ is measured by the probability that an unseen sample $(\tilde{x},f(\tilde{x}))\in G(f)$ is not contained in $E(f)$, i.e., $P[(\tilde{x},f(\tilde{x}))\notin E(f)]$. By Algorithm~\ref{SVP}, we can guarantee with high confidence that this probability is less than $\bar{\mu}$ as shown in the next theorem.

\begin{thm}[Probabilistic accuracy of $E(f)$]\label{ThmProbabilisticAcc}
	Let $\tilde{x}_i$ and $\tilde{d}_i,i=S+1,\dots,S+V,$ from validation data \eqref{DataSetVal} be independently sampled from distribution $\Omega_x$ over $\mathbb{X}$ and distribution $\Omega_d$ over $\mathcal{B}_{\epsilon}(0)$, respectively. Then	$P[(\tilde{x},f(\tilde{x}))\notin E(f)]\leq\tilde{\mu}+\epsilon_c$ with confidence at least $1-c$ for $\tilde{x}$ sampled from $\Omega_x$.
\end{thm}
\begin{proof}
	To conclude on $P[(\tilde{x},f(\tilde{x}))\notin E(f)]$ from empirical sampling, we apply Hoeffding’s inequality as, e.g., in \cite{Hertneck}. To this end, let $\tilde{d}$ be sampled from $\Omega_d$ over $\mathcal{B}_{\epsilon}(0)$. Then, for the indicator function $I$ from \eqref{IndicatorFcn},
	\begin{align}\label{ProbIneq}
	P[(\tilde{x},f(\tilde{x}))\notin E(f)]\leq P[I(\tilde{x},f(\tilde{x})+\tilde{d})=1]
	\end{align}  
	because $f(\tilde{x})\in\mathcal{B}_{\epsilon}(f(\tilde{x})+\tilde{d})$. Furthermore, since $\tilde{x}_i$ and $\tilde{d}_i$ are each independently sampled and identically distributed (iid), $\tilde{y}_i$ are iid, and thus $I(\tilde{x}_i,\tilde{y}_i)$ as well. Together with $\tilde{y}_i=f(\tilde{x}_i)+\tilde{d}_i$, Hoeffding’s inequality \cite{Hoeffding}
	\begin{equation*}
	P\left[\Big|\Big|\tilde{\mu}-P[I(\tilde{x},f(\tilde{x})+\tilde{d})=1]\Big|\Big|_2\geq\epsilon_c\right]\leq 2\,\text{exp}(-2V\epsilon_c^2),
	\end{equation*}
	with the empirical risk $\tilde{\mu}$ from \eqref{EmpRisk}, implies that
	\begin{equation*}
	P[I(\tilde{x},f(\tilde{x})+\tilde{d})=1]\leq\tilde{\mu}+\epsilon_c
	\end{equation*}	
	with confidence at least $1-c$. Together with \eqref{ProbIneq}, the claim is proven.
\end{proof}

In Algorithm~\ref{SVP}, if $\tilde{\mu}+\epsilon_c$ is larger than the chosen critical risk $\bar{\mu}$, then the envelope $E(f)$ does probably not contain the whole graph of $f$, and hence $M_{i,\alpha}$ must be increased. Since we assume for verifying system properties in the subsequent section that the envelope $E(f)$ includes $G(f)$, $\tilde{\mu}+\epsilon_c$ also corresponds to the risk that the determined system properties does not hold.  
\begin{rmk}[Extensions of Algorithm~\ref{SVP}]\label{RmkChoiceofk}
	Algorithm~\ref{SVP} can be adapted by a multidimensional bisection if equal $M_{i,\alpha}$ are too conservative. Furthermore, if $k$ is not known, then Algorithm~\ref{SVP} can be extended by an iteration over $k$. To this end, we perform Algorithm~\ref{SVP}, e.g., for linear TP approximation by setting $k=1$ and then increase $k$ until an envelope with small risk $\bar{\mu}$ can still be constructed. The algorithm can also be considered for various TPs in order to exploit envelope $E'(f)$.	
\end{rmk}
\begin{rmk}[$M_{i,\alpha}$ as hyperparameter]
	$M_{i,\alpha}$ can be seen as hyperparameter of the data-driven representation $E(f)$. Hyperparameters are typically chosen such that the likelihood to observe the available data is maximized. Similarly, here $M_{i,\alpha}$ minimize the risk that an unseen sample is not contained in $E(f)$. Moreover, hyperparameters are commonly validated using validation procedures \cite{MachineTS} (Section 4.2) similar to Algorithm~\ref{SVP}.
\end{rmk}

\section{Data-driven dissipativity verification for nonlinear systems}\label{SecDissi}

\begin{figure*}
	\setcounter{MYtempeqncnt}{\value{equation}}
	\setcounter{equation}{27}
	\begin{align}\label{SOSCond}
	\Psi(x,u)=\Upsilon^T
	\text{diag}\left(\begin{bmatrix}0 & -\mathcal{X}\frac{\partial m(x)}{\partial x}\\ -\frac{\partial m(x)}{\partial x}^T\mathcal{X} & 0\end{bmatrix}\Bigg| s(x,u)\Bigg|\sum_{i=1}^{n_p}t_i(x,u)p_i(x,u)\bigg|\tau_{\text{sec}}(x,u)\Phi(\omega)\bigg|\tau_{\text{sm}}(x,u)\varDelta_*
	\right)\Upsilon
	\end{align}
	\setcounter{equation}{\value{MYtempeqncnt}}
	\hrulefill
\end{figure*}
In this section, we apply the data-based envelopes $E(f)$ and $E'(f)$ for nonlinear functions to develop a framework to check whether a nonlinear dynamical system is dissipative. \\\indent
Related to the setup in Section~\ref{SecProblemSetup}, we examine the unknown nonlinear continuous-time system
\begin{align}\label{TrueSystem}
\dot{x}(t)=f(x(t),u(t))=\begin{bmatrix}\vspace{0.01cm}\begin{matrix}
f_1(x(t),u(t))\vspace{-0.1cm}\\\vdots\\f_{n_x}(x(t),u(t))\end{matrix}\end{bmatrix} 
\end{align}
within the compact and convex operation set 
\begin{equation}\label{Constraints}
\begin{aligned}
\mathbb{P}{=}\{(x,u)\in\mathbb{R}^{n_x}\times\mathbb{R}^{n_u}&: p_i(x,u)\leq 0,\  p_i\in\mathbb{R}[x,u],\\ &\hspace{2.0cm}i=1,\dots,n_p\}.
\end{aligned}
\end{equation}
We assume $f:\mathbb{P}\rightarrow\mathbb{R}^{n_x}$ is $k+1$ times continuously differentiable and Assumption~\ref{AssBoundDeri} is satisfied. Moreover, let $\mathbb{X}\subset\mathbb{R}^{n_x}$ denote the projection of ${\mathbb{P}}$ on $\mathbb{R}^{n_x}$. We suppose the access to measurements
\begin{equation}\label{DataSys}
\{(\dot{\tilde{x}}_i,\tilde{x}_i,\tilde{u}_i)_{i=1,\dots,S}\}
\end{equation}
in presence of noise, i.e., $\dot{\tilde{x}}_i=f(\tilde{x}_i,\tilde{u}_i)+\tilde{d}_i$ where $\tilde{d}_1,\dots,\tilde{d}_S$ satisfy Assumption~\ref{AssNoiseBound}. The disturbance $\tilde{d}_i$ cover inaccurate estimations of $\dot{\tilde{x}}_i$, potential effects of noisy measurements of $\tilde{x}_i$ on the estimation of $\dot{\tilde{x}}_i$, and process noise. The samples can be measured from a single trajectory of system~\eqref{TrueSystem} or from multiple trajectories with even varying sampling rates. In the sequel, we aim to determine dissipativity properties for system~\eqref{TrueSystem} from the noisy input-state-velocity data \eqref{DataSys}.   

\begin{defn}[Dissipativity]\label{DissiDef}
	System~\eqref{TrueSystem} is dissipative on $\mathbb{P}$ with respect to the supply rate $s:\mathbb{P}\rightarrow\mathbb{R}$ if there exists a continuous storage function $\lambda:\mathbb{X}\rightarrow[0,\infty)$ such that
	\begin{align}\label{dissipativityInquInt}
	\lambda(x(T))\leq 	\lambda(x(0))+\int_{0}^{T}s(x(t),u(t))\text{d}t 
	\end{align}
	for all $T\geq0$ and all inputs $u:[0,T]\rightarrow\mathbb{R}^{n_u}$ such that $(x(t),u(t))\in\mathbb{P},t\in[0,T]$. If $\lambda(x)$ is continuously differentiable, then this is implied by
	\begin{align}\label{dissipativityInqu}
	\frac{\partial{\lambda(x)}}{\partial x}f(x,u)\leq s(x,u),\quad \forall (x,u)\in{\mathbb{P}}.
	\end{align}
\end{defn} 
\vspace{0.5cm}

We assess a regional rather than a global analysis of dissipativity as an advantage because $\mathbb{P}$ can be chosen as the physical constraints of states and inputs. Therefore, regional dissipativity provides more precise insights into the system behaviour. Moreover, since we do not assume knowledge on the structure of the dynamics, a global model cannot be learned from data or would be conservative. Even using first principles, an accurate global model for a nonlinear system is typically not available. Also Assumption~\ref{AssBoundDeri} can be restrictive for $(x,u)\in\mathbb{R}^{n_x}\times\mathbb{R}^{n_u}$. \\\indent
As an example for a controller design based on a regional dissipativity property, let a system be passive on $\mathbb{X}\times\mathbb{R}^{n_u}$ for a positive definite storage function. Although there exist open-loop trajectories leaving $\mathbb{X}$ and hence may not satisfy the passivity property, we can use the regional passivity for the design of an output feedback $u=-Ky,K\succeq0,$ to stabilize the system. Indeed, the largest sublevel-set of the storage function within $\mathbb{X}$ is invariant w.r.t. the closed-loop dynamics for any $K\succeq0$. Hence, the closed-loop trajectories stay in $\mathbb{X}$ and satisfy the passivity condition. Since the controller achieves invariance of $\mathbb{X}$, we only verify dissipativity for trajectories staying in $\mathbb{X}$ without guaranteeing that open-loop solutions stay in $\mathbb{X}$. To achieve invariance by a controller, control invariance of $\mathbb{X}$ is necessary, i.e., for all initial conditions $x_0=x(0)\in\mathbb{X}$ and $T\geq0$ there exists an input $u:[0,T]\rightarrow\mathbb{R}^{n_u}$ such that $(x(t),u(t))\in\mathbb{P},t\in[0,T]$. Control invariance can be verified from data by determining a regionally stabilizing state feedback \cite{Gaussian_noise} (Section 6.2).\\\indent
While determining dissipativity for known nonlinear systems using Taylor approximation is examined thoroughly in \cite{TaylorPassivity}, we extend this result to a data-driven method exploiting the envelope $E(f)$. Since $f$ is contained in $E(f)$, the ground-truth system \eqref{TrueSystem} is dissipative if all $\left(\xi,\dot{x}\right)\in E(f)$ with $\xi=\begin{bmatrix}x^T&u^T\end{bmatrix}^T$ satisfy the dissipativity criterion \eqref{dissipativityInqu}.

\begin{thm}[Dissipativity verification by $E(f)$]\label{ThmDissi}
	Let Assumption~\ref{AssBoundDeri} and \ref{AssNoiseBound} hold. Moreover, let the matrix $\tilde{Z}=\begin{bmatrix}z(\tilde{x}_1,\tilde{u}_1) & \cdots & z(\tilde{x}_S,\tilde{u}_S)\end{bmatrix}$ have full row rank for the TP $T_k(\omega)[f(x,u)]=A^*z(x,u)$ at $\omega=[\omega_x^T\ \omega_u^T]^T$. Then the set-membership 
	\begin{equation}\label{Setmem}
	\tilde{\Sigma}_{\omega}=\left\{A:\begin{bmatrix}I_{n_{z}}\\A\end{bmatrix}^T\varDelta_{*}\begin{bmatrix}I_{n_{z}}\\A\end{bmatrix}\preceq0\right\}
	\end{equation}
	for $A^*$ exists. Moreover, system~\eqref{TrueSystem} is dissipative on \eqref{Constraints} regarding the supply rate $s\in\mathbb{R}[x,u]$ if there exist a storage function $m(x)^T\mathcal{X}m(x)$ with $\mathcal{X}\succeq0$ and $m\in\mathbb{R}[x]^{n_m}$ and polynomials $t_1,\dots,t_{n_p},\tau_{\text{sec}},\tau_{\text{sm}}\in\text{SOS}[x,u]$ such that $\Psi(x,u)$ in \eqref{SOSCond} is an SOS matrix with
	\begin{equation*}
	\Upsilon = \begin{bmatrix}\hspace{0.2cm}\begin{matrix}	 0 & 0 & m(x) \\ I_{n_x} & 0 & 0\\	\hline 
	0 & 0 & 1\\\hline 
	0 & 0& 1\\\hline I_{n_x} & -I_{n_x} & 0  \\  0 & 0 & 1\\\hline 0 & 0 & z(x,u)\\ 0 & I_{n_x} & 0\end{matrix} \hspace{0.2cm}\end{bmatrix}.	
	\end{equation*}
	\setcounter{equation}{28}
\end{thm}
\begin{proof}
	\eqref{Setmem} exists by Proposition~\ref{Primal_elli}. Furthermore, due to the fact that any SOS matrix is positive semidefinite, 
	$0\leq\sigma^T\Psi\sigma$ with $\sigma=\begin{bmatrix}\dot{x}^T &  (Az(x,u))^T & 1\end{bmatrix}^T$ and
	\begin{align*}
	\sigma^T\Psi\sigma=&\, -m^T\mathcal{X}\frac{\partial m}{\partial x}\dot{x}-\dot{x}^T\frac{\partial m}{\partial x}^T\mathcal{X}m+s+\sum_{i=1}^{n_p}t_ip_i\\	
	&+\tau_{\text{sec}}\star^T
	\Phi \cdot\begin{bmatrix}\begin{array}{c}\begin{matrix}
	I_{n_x} & -I_{n_x} & 0 \\\hline 0 & 0 & 1\end{matrix}\end{array}\end{bmatrix}\begin{bmatrix}\dot{x}\\ Az\\1	\end{bmatrix}\\
	&+\tau_{\text{sm}}\star^T
	\varDelta_*\cdot
	\begin{bmatrix}I_{n_z}\\ A	\end{bmatrix}z.
	\end{align*}
	Since $p_i(x,u)\leq0$ by \eqref{Constraints}, $p_\text{sec}(\xi,\dot{x},A)\leq0$ by \eqref{Envelop1}, and $\star^T\varDelta_*\cdot\begin{bmatrix}I_{n_z}& A^T	\end{bmatrix}^T\preceq0$ by \eqref{Setmem}, the generalized S-procedure of Proposition~\ref{SOSRelaxation} implies 
	\begin{align}\label{DissiInqu}
	0\leq {-}m(x)^T\mathcal{X}\frac{\partial m(x)}{\partial x}\dot{x}{-}\dot{x}^T\frac{\partial m(x)}{\partial x}^T\mathcal{X}m(x){+}s(x,u)
	\end{align}
	for all $\left(\xi,\dot{x}\right)\in E(f)$. Given that \eqref{DissiInqu} corresponds to \eqref{dissipativityInqu} for $\lambda(x)=m(x)^T\mathcal{X}m(x)$ and $G(f)\subseteq E(f)$, the claim is proven. 
\end{proof}

Note that the rank condition of $\tilde{Z}$ depends on the TP as $z(x,u)$ includes $\omega$. In view of \cite{Ellipsoid_DePersis}, this rank condition can be seen as a persistence of excitation condition of the data. Since $\Psi$ is linear regarding $\mathcal{X},t_1,\dots,t_{n_p},\tau_{\text{sec}}$, and $\tau_{\text{sm}}$, these parameters can be optimized to render $\Psi$ an SOS matrix using standard SOS optimization. The vector $m(x)$ contains usually monomials of $x$ where their degree can iteratively be increased to refine the system property. As indicated in Section~\ref{Sec_Improv}, it might be beneficial to exploit prior knowledge on $f$ by polynomial sector bounds as \eqref{PolySecf1} and set-memberships as \eqref{SetMemfj}. These can analogously be employed in Theorem~\ref{ThmDissi} by the generalized S-procedure. We finally summarize the data-driven verification of dissipativity from data in Algorithm~\ref{Algorithm2}.

\begin{algo}[Data-driven dissipativity verification]\label{Algorithm2}\indent
	\begin{itemize}	
		\item[$0.)$] Given data \eqref{DataSys} such that $\tilde{Z}$ has full row rank and the knowledge that $f$ from \eqref{TrueSystem} is at least $k+1$ times continuously differentiable. Choose $\omega\in\mathbb{R}^{n_x+n_u}$.
		\item[$1.)$] If bounds $M_{i,\alpha}$ are not available, then follow Algorithm~\ref{SVP} by splitting data \eqref{DataSys} into data to compute $E(f)$ and validation data.			
		\item[$2.)$] Compute $\Xi_\text{p}$ in Proposition~\ref{Primal_elli} by solving LMI~\eqref{Elli_LMI} and compute $\tilde{\Sigma}_\omega$ as in \eqref{Sigma_A}.		
		\item[$3.)$] For a chosen $m(x)$, find $\mathcal{X}\succeq0$ and polynomials $t_1,\dots,t_{n_p},\tau_{\text{sec}},\tau_{\text{sm}}\in\text{SOS}[x,u]$ such that $\Psi(x,u)$ is an SOS matrix. 
	\end{itemize}
\end{algo}

\begin{figure*}
	\setcounter{MYtempeqncnt}{\value{equation}}
	\setcounter{equation}{32}
	\begin{align}\label{SOSCond1}
	\Psi_j=\tilde{\Upsilon}^T
	\text{diag}\left(\begin{bmatrix}0 & -\mathcal{X}_j\frac{\partial m}{\partial x}\\ -\frac{\partial m}{\partial x}^T\mathcal{X}_j & 0\end{bmatrix}\Bigg| s\Bigg|\sum_{i=1}^{n_v}r_{ji}v_i\bigg|\sum_{i=1}^{n_{jc}}t_{ji}c_{ji}\bigg|\tau_{j1}\Phi(\omega_1)\bigg|\cdots\bigg|\tau_{jn_\omega}\Phi(\omega_{n_\omega})\bigg|\rho_{j1}\varDelta_{*,1}\bigg|\cdots\bigg|\rho_{jn_\omega}\varDelta_{*,n_\omega}
	\right)\tilde{\Upsilon}
	\end{align}
	\setcounter{equation}{\value{MYtempeqncnt}}
	\hrulefill
\end{figure*}

\subsection{Dissipativity verification by E\textquotesingle(f)}

We show next how the combination of multiple envelopes from TPs around various points \eqref{InterTSE} can be incorporated into Theorem~\ref{ThmDissi} to reduce the conservatism of the data-based dissipativity verification. For that purpose, $\Psi$ can be extended by applying the generalized S-procedure as in Theorem~\ref{ThmDissi} for the polynomial sector bound \eqref{Envelop1} and set-membership \eqref{Sigma_A} of each TP. However, due to the conservatism of the S-procedure, the problem emerges that only one envelope is taken into account as already observed in \cite{MartinGraphAppro}. This issue can be circumvented by means of a partition of $\mathbb{X}$ as common for switched systems \cite{SwitchedSys}. To this end, let the operation set $\mathbb{P}=\mathbb{X}\times\mathbb{U}$ with $\mathbb{U}=\{u: v_{i}(u)\leq 0,\  v_{i}\in\mathbb{R}[u],i=1,\dots,n_{v}\}$. Additionally, we consider a partition of $\mathbb{X}$ into subsets $\mathbb{X}_j,j=1,\dots, n_\text{part}$, each described by polynomial inequalities $\mathbb{X}_j=\{x: c_{ji}(x)\leq 0,\  c_{ji}\in\mathbb{R}[x],i=1,\dots,n_{jc}\},$ with $\mathbb{X}=\bigcup_{j=1}^{n_\text{part}}\mathbb{X}_j$, $\text{int}\left(\mathbb{X}_j\right)\cap\text{int}\left(\mathbb{X}_l\right)=\emptyset,j\neq l$. Here $\text{int}$ denotes the interior of a set. The boundary between $\mathbb{X}_j$ and $\mathbb{X}_l, j\neq l,$ is given by $S_{jl}=\{x:h_{jl0}(x)=0,h_{jlm}(x)\leq0,m=1,\dots,n_{jh}\},$ which might be empty. Then dissipativity can be determined for system~\eqref{TrueSystem} from data by $E'(f)$.

\begin{coro}[Dissipativity verification by $E'(f)$]\label{ThmDissi2}
	Let Assumption~\ref{AssBoundDeri} and \ref{AssNoiseBound} hold. Moreover, let the matrices $\tilde{Z}_\ell=\begin{bmatrix}z_\ell(\tilde{x}_1,\tilde{u}_1) & \cdots & z_\ell(\tilde{x}_S,\tilde{u}_S)\end{bmatrix}$ for the TPs $T_k(\omega_\ell)[f(x,u)]=A_\ell^* z_\ell(x,u)$ around $\omega_\ell=[\omega_{x,\ell}^T\ \omega_{u,\ell}^T]^T,\ell=1,\dots,n_\omega,$ have full row rank. Then set-memberships 
	\begin{equation}\label{SetMems}
	\tilde{\Sigma}_{\omega_{\ell}}=\left\{A_\ell:\begin{bmatrix}I_{n_{z_\ell}}\\A_\ell\end{bmatrix}^T\varDelta_{*,\ell}\begin{bmatrix}I_{n_{z_\ell}}\\A_\ell\end{bmatrix}\preceq0\right\}
	\end{equation}
	for $A_\ell^*$ exist. Moreover, system~\eqref{TrueSystem} is dissipative on $\mathbb{X}\times\mathbb{U}$ regarding the supply rate $s\in\mathbb{R}[x,u]$ if there exist for $j=1,\dots,n_\text{part}$ a piecewise-defined storage function
	\begin{equation*}
	\lambda(x)=m(x)^T\mathcal{X}_jm(x)\in\mathbb{R}[x],\mathcal{X}_j\in\mathbb{R}^{n_m\times n_m},\text{for\ } x\in\mathbb{X}_j,
	\end{equation*}
	with $\lambda(0)=0$ and $m\in\mathbb{R}[x]^{n_m}$, and polynomials $s_{jl}\in\mathbb{R}[x],l=1,\dots,n_\text{part},l\neq j$, $r_{ji}\in\text{SOS}[x,u],i=1,\dots,n_v$, $t_{ji}\in\text{SOS}[x,u],i=1,\dots,n_{jc}$, $\tau_{ji},\rho_{ji}\in\text{SOS}[x,u],i=1,\dots,n_\omega$, and $\tau_{\text{pos},ji}\in\text{SOS}[x],i=1,\dots,n_{jc}$ such that 
	\begin{align}
	&m^T\mathcal{X}_jm+\sum_{i=1}^{n_{jc}}\tau_{\text{pos},ji}c_{ji}\in\text{SOS}[x],\label{Pos}\\
	&m^T\mathcal{X}_jm+s_{jl}h_{jl0}-m^T\mathcal{X}_lm=0,\label{Cont}
	\end{align}
	for $l=1,\dots,n_\text{part},l\neq j$ and $S_{jl}\neq\emptyset$, and $\Psi_j(x,u)$ in \eqref{SOSCond1} are SOS matrices with
	\begin{equation*}
	\tilde{\Upsilon} = \begin{bmatrix}\hspace{0.2cm}\begin{matrix}	 0 & 0 & \cdots & 0 & m(x) \\ I_{n_x} & 0 & \cdots & 0& 0\\	\hline 
	0 & 0& \cdots & 0 & 1\\\hline  0 & 0& \cdots & 0 & 1\\\hline 
	0 & 0& \cdots & 0& 1\\\hline I_{n_x} & -I_{n_x}& \cdots & 0 & 0  \\  0 & 0& \cdots & 0 & 1\\\hline \vdots & \vdots & \ddots & \vdots & \vdots \\\hline I_{n_x} & 0& \cdots & -I_{n_x} & 0  \\  0 & 0& \cdots & 0 & 1\\\hline 0 & 0& \cdots & 0 & z_1(x,u)\\ 0 & I_{n_x}& \cdots & 0 & 0\\\hline \vdots & \vdots & \ddots & \vdots & \vdots\\\hline 0 & 0& \cdots & 0 & z_{n_\omega}(x,u)\\ 0 & 0& \cdots & I_{n_x} & 0 \end{matrix} \hspace{0.2cm}\end{bmatrix}.	
	\end{equation*}
	\setcounter{equation}{33}
\end{coro}
\begin{proof}
	Set-memberships \eqref{SetMems} exist by Proposition~\ref{Primal_elli}. By \eqref{Pos}, Proposition~\ref{SOSRelaxation} implies $\lambda(x)\geq0$ for all $x\in\mathbb{X}$. Equation \eqref{Cont} ensures that $\lambda(x)$ is continuous as shown in \cite{SwitchedSys}. For that reason, it remains to prove that \eqref{dissipativityInquInt} holds for all trajectories within $\mathbb{P}$.\\\indent 
	To this end, we proceed as in the proof of Theorem~\ref{ThmDissi}, but multiplying $\Psi_j$ from both sides with $\begin{bmatrix}\dot{x}^T &  (A_1z_1(x,u))^T & \cdots &(A_{n_\omega}z_{n_\omega}(x,u))^T & & 1\end{bmatrix}^T$. The generalized S-procedure implies
	\begin{align}\label{DissiInequ}
	m(x)^T\mathcal{X}_j\frac{\partial m(x)}{\partial x}\dot{x}+\dot{x}^T\frac{\partial m(x)}{\partial x}^T\mathcal{X}_jm(x)\leq s(x,u)
	\end{align}
	for all $(x,u)\in\mathbb{X}_j\times\mathbb{U}$ and $\left(\xi,\dot{x}\right)\in E'(f)$. Since $E'(f)$ contains the graph of $f(x,u)$ and \eqref{DissiInequ} corresponds to \eqref{dissipativityInqu}, $m(x)^T\mathcal{X}_jm(x)$ satisfies \eqref{dissipativityInquInt} for all trajectories of \eqref{TrueSystem} within $\mathbb{X}_j\times\mathbb{U}$. To conclude that $\lambda(x)$ is a storage function satisfying \eqref{dissipativityInquInt} for all trajectories of \eqref{TrueSystem} within $\mathbb{X}\times\mathbb{U}$, we study exemplarily a trajectory with transition between two subsets
	\begin{equation*}
	(x(t),u(t))\in\begin{cases}
	\text{int}\left(\mathbb{X}_{j}\right)\times\mathbb{U},\quad &t_0\leq t< t_1\\S_{jl}\times\mathbb{U},\quad &t_1\leq t< t_2\\\text{int}\left(\mathbb{X}_{l}\right)\times\mathbb{U},\quad &t_2\leq t\leq t_3
	\end{cases}.
	\end{equation*}
	Since $(x(t),u(t))\in\mathbb{X}_l\times\mathbb{U},t_1\leq t\leq t_3,$ and $m(x)^T\mathcal{X}_{l}m(x)$ satisfies \eqref{dissipativityInquInt} for all trajectories of \eqref{TrueSystem} within $\mathbb{X}_{l}\times\mathbb{U}$, then
	\begin{align*}
	\star^T\mathcal{X}_{l}\cdot m(x(t_3))\leq \star^T\mathcal{X}_{l}\cdot m(x(t_1))+\int_{t_1}^{t_3}s(x(t),u(t))\text{d}t.
	\end{align*}
	According to $m(x)^T\mathcal{X}_{l}m(x)=m(x)^T\mathcal{X}_{j}m(x)$ for $x\in S_{jl}$ and 
	\begin{align*}
	\star^T\mathcal{X}_{j}\cdot m(x(t_1))\leq \star^T\mathcal{X}_{j}\cdot m(x(t_0))+\int_{t_0}^{t_1}s(x(t),u(t))\text{d}t,
	\end{align*}
	it holds that 
	\begin{align*}
	\star^T\mathcal{X}_{l}\cdot m(x(t_3))\leq \star^T\mathcal{X}_{j}\cdot m(x(t_0))+\int_{t_0}^{t_3}s(x(t),u(t))\text{d}t.
	\end{align*}		
	Since $\star^T\mathcal{X}_{l}\cdot m(x(t_3))=\lambda(x(t_3))$ and $\star^T\mathcal{X}_{j}\cdot m(x(t_0))=\lambda(x(t_0))$ and $j$ and $l$ are arbitrary, $\lambda(x)$ is a storage function satisfying \eqref{dissipativityInquInt} for all trajectories of \eqref{TrueSystem} within $\mathbb{X}\times\mathbb{U}$.  
\end{proof}

Note that the rank condition of $\tilde{Z}_{\ell},{\ell}=1,\dots, n_\omega,$ must be checked for each TP separately as the polynomial vectors $z_{\ell}(x,u)$ depend on $\omega_{\ell}$. Furthermore, the TPs usually provide a good approximation of the nonlinear dynamics only close to $\omega_\ell$. Thereby, we could also employ for each $\Psi_j(x,u)$ only one envelope from a TP at $\omega_\ell\in\mathbb{X}_j\times\mathbb{U}$ instead of all envelopes. Thus, only multiplier $\tau_{j}$ and $\rho_{j}$ for $j=1,\dots,n_\text{part}$ instead of $\tau_{j\ell},\rho_{j\ell}\in\text{SOS}[x,u],{\ell}=1,\dots,n_\omega,$ would be required such that the computational complexity of the SOS optimization problem in Corollary~\ref{ThmDissi2} is reduced. To further reduce the computational complexity, a joint storage function with $\mathcal{X}=\mathcal{X}_1=\dots=\mathcal{X}_{n_\text{part}}\succeq0$ is conceivable. The data-driven verification of dissipativity with Corollary~\ref{ThmDissi2} pursues Algorithm~\ref{Algorithm2} but computing multiple set-memberships \eqref{SetMems} in step 2.) and checking the SOS conditions in Corollary~\ref{ThmDissi2} in step 3.). \\\indent
While $E'(f)$ can improve the inference on the dissipativity property, the choice of $\omega_{\ell}$ is decisive but often unclear a priori. To this end, we suggest an iterative procedure to determine a meaningful choice of $\omega_{\ell}$ in the next remark.  
\begin{rmk}[Choice of $\omega_{\ell}$]\label{RmkOmega}
	For the sake of explanation, let Corollary~\ref{ThmDissi2} be applied to determine an upper bound on the $\mathcal{L}_2$-gain of a system on $\mathbb{X}\times\mathbb{U}$. The $\mathcal{L}_2$-gain corresponds to the minimal $\gamma> 0$ for the dissipativity with supply rate $s(x,u)=\gamma^2 u^Tu-x^Tx$. Given the solution of Corollary~\ref{ThmDissi2} for a partition of $\mathbb{X}$ and one TP, e.g., at the center of $\mathbb{X}$. Then we can use the obtained storage function to determine in which of the subsets the system property is most conservative, i.e., the subset with the largest $\gamma$. Thus, reducing the conservatism within this subset improves the inference on the system property. Hence, we determine a second TP with $\omega_{\ell}$ within the identified subset. This procedure can then be repeated. Note that this procedure circumvents to some extent the problem that verifying system properties by SOS optimization does not provide the state-input pair where \eqref{dissipativityInqu} is tight.
\end{rmk}

\subsection{$\mathcal{L}_2$-gain of a two-tank system}\label{SecExTank}

In our preliminary work \cite{MartinNL}, we infer for two numerical examples a guaranteed upper bound on the $\mathcal{L}_2$-gain. Contrary, we study here an experimental example.\\\indent
We analyze the $\mathcal{L}_2$-gain of the two-tank system in Fig.~\ref{Fig.2Tank}.
\begin{figure}
	\centering
	\includegraphics[width=0.49\linewidth]{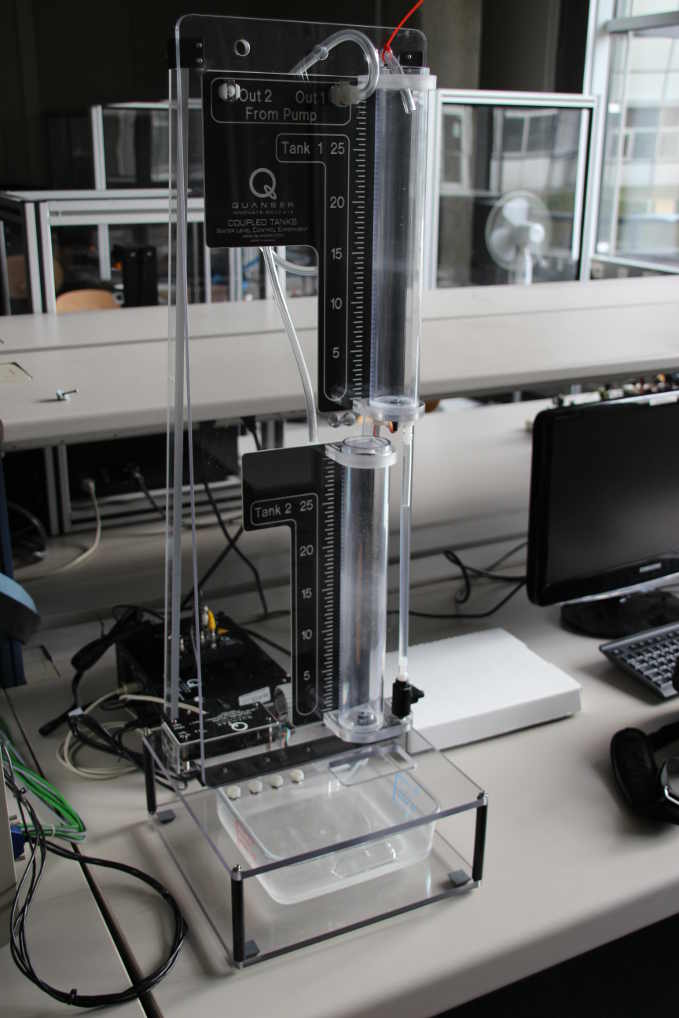}
	\includegraphics[width=0.49\linewidth]{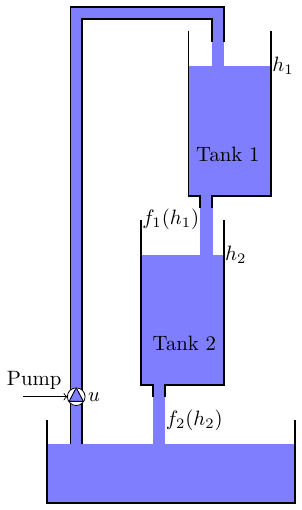}
	\caption{Experimental setup of the two-tank system.}
	\label{Fig.2Tank}
\end{figure}
In view of the experimental setup, it is clear that the input $u$ directly affects only the water height of the first tank $h_1$ regardless of $h_1$ and the water height of the second tank $h_2$. Moreover, the outflows $f_1$ and $f_2$ only depend on the corresponding water heights $h_1$ and $h_2$, respectively. The outflow of the first tank is equal to the inflow of the second tank. In terms of these observations, we suppose 
\begin{align*}
\dot{h}_1&=-f_1(h_1)+c_u u+d_1, h_1\in[0.05\,\text{m}, 0.25\,\text{m}], \\
\dot{h}_2&=f_1(h_1)-f_2(h_2)+d_2, h_2\in[0.05\,\text{m}, 0.25\,\text{m}]
\end{align*}
as prior knowledge on the structure of the system dynamics with unknown functions $f_1$ and $f_2$, unknown parameter $c_u$, and disturbances $d_1$ and $d_2$. We will exploit this knowledge as described in Section~\ref{Sec_Improv} to infer on the $\mathcal{L}_2$-gain of the two-tank system, i.e., we are interested in the minimal $\gamma> 0$ such that the system is dissipative with respect to the supply rate $\gamma^2 u^Tu-\begin{bmatrix}h_1 & h_2\end{bmatrix}\begin{bmatrix}h_1 & h_2\end{bmatrix}^T$.\\\indent
We measure the four input-state trajectories in Fig.~\ref{Fig.Experiment} from four open-loop experiments.
\begin{figure}
	\centering
	\includegraphics[width=1\linewidth]{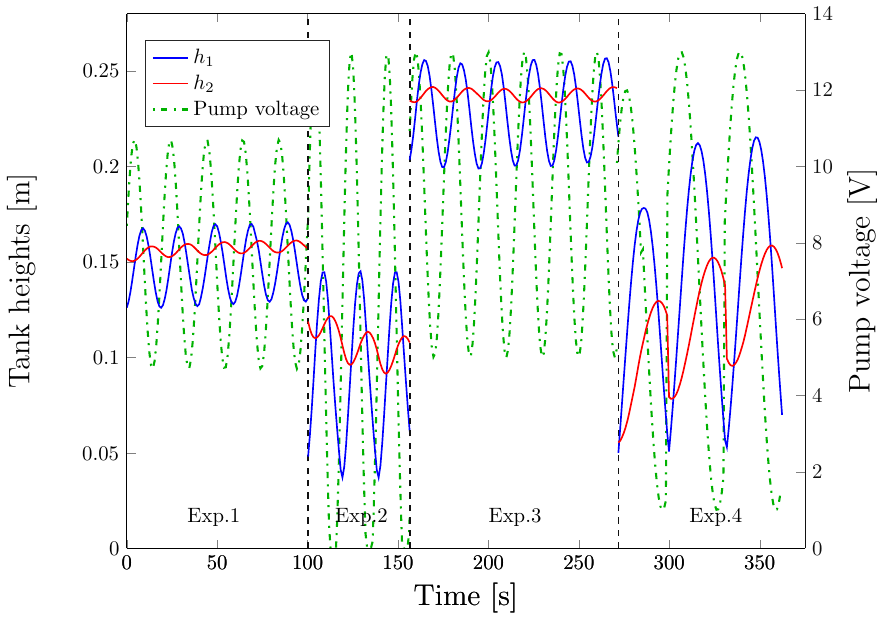}
	\caption{Input-state trajectories from four open-loop experiments.}
	\label{Fig.Experiment}
\end{figure}
Note that the samples outside of the operation set are omitted. While the sampling time of the sensors is $2\,\text{ms}$, the system exhibits a slower dynamics. Thus, to derive time sequences of $\dot{h}_1$ and $\dot{h}_2$, we first smooth the measured state signals by averaging over a window of 400 samples before and after each time point. Moreover, we use a triangular weighting, i.e., points in the middle have a high weighting while the weighting decreases linearly to zero towards the window edges. The smoothed signals are then leveraged to obtain the time signal of $\dot{h}_1$ and $\dot{h}_2$ by central finite difference with sixth order of accuracy. In contrast to \cite{AnneDissi}, we do not require an additional experiment to determine an equilibrium for linearizing the system dynamics as we do not need to set the $0$-th order elements of the TPs to zero. Fixing an operation point is also difficult as the measured equilibrium for the same input varies between experiments.\\\indent
To derive data-driven inferences on the $\mathcal{L}_2$-gain, we follow Algorithm~\ref{Algorithm2}. To compute $E(f)$, we extract $50$ input-state-velocity snapshots equally distributed over time from each of the smoothed trajectories of experiment 1 to 3. Moreover, we derive approximately independent validation data from a single trajectory as described in Section 2.B of \cite{IID_data}. For that purpose, we consider $1000$ snapshots from the trajectory of the fourth experiment equally distributed over time as validation data \eqref{DataSetVal}. Since the bounds on the $k+1$-th order partial derivatives 
$M_{1,[k+1\ 0\ 0]}=M_{2,[k+1\ 0\ 0]},$ and $ M_{2,[0\ k+1\ 0]}$ as well as noise bounds $\epsilon_i$ with $||d_i||_2\leq \epsilon_i,i=1,2,$ are obscure, we modify Algorithm~\ref{SVP}: for $M_{1,[k+1\ 0\ 0]}=M_{2,[k+1\ 0\ 0]}=M_{2,[0\ k+1\ 0]}=\beta_M\bar{M}_{k+1},\beta_M\in[0,1]$ and $\bar{M}_{k+1}$ given in Table~\ref{MBounds}, the minimal noise bounds $\epsilon_i, i=1,2,$ are chosen subject to $\tilde{\Sigma}_\omega$ is not empty and the resulting envelope $E(f)$ is consistent with all $1000$ validation samples, i.e., $\tilde{\mu}=0$. 
Hence, the risk that an other data point from the fourth experiment violates the data-driven envelopes is smaller than $0.04$ with confidence $0.9$ by Theorem~\ref{ThmProbabilisticAcc}. The risk only holds for samples from the fourth experiment as the validation data do not cover all possible state-input pairs. Nevertheless, this risk is meaningful as the validation data cover typical state-input trajectories.\\\indent
Based on the obtained bounds on the $k+1$-th order partial derivatives and the noise, we conclude by Theorem~\ref{ThmDissi} on the $\mathcal{L}_2$-gain for a quartic storage function, i.e., $m(h_1,h_2)=\begin{bmatrix}h_1 & h_2 & h_1^2 & h_2^2\end{bmatrix}^T$, and for TPs at $\omega=\begin{bmatrix}0.15 & 0.15 & 0\end{bmatrix}^T$ of order $k=1,2,3$.
\begin{figure}
	\centering
	\includegraphics[width=1\linewidth]{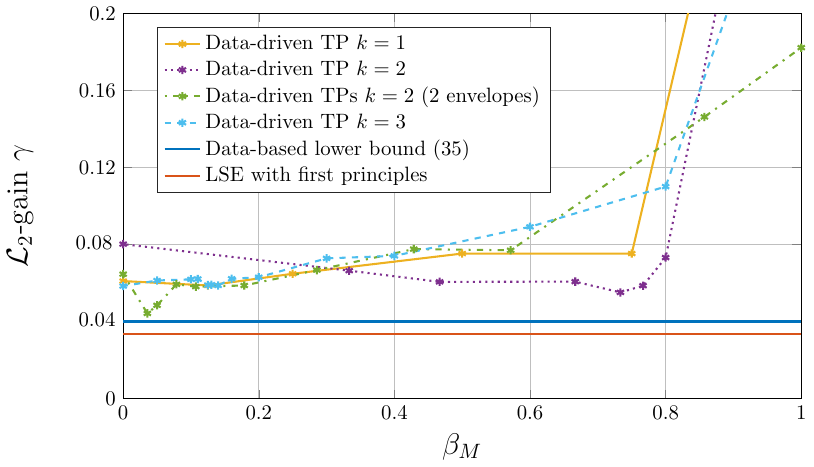}
	\caption{Inference by Theorem~\ref{ThmDissi} and Corollary~\ref{ThmDissi2} for TPs of order $k=1,2,3$, inference by LSE, and		data-driven lower bound \eqref{LowerBound}. }
	\label{Fig.L2GainTank}
\end{figure}
\begin{table}
	\begin{center}
		\caption{Maximal bounds on $k+1$-th order partial derivatives}\label{MBounds}
		\begin{tabular}{c||c|c|c|c}
			$k$ & 1 & 2 & 2 (2 envelops) & 3 \\\hline
			$\bar{M}_{k+1}$\phantom{\Big|} & 1 & 15 & 140 & 1000
		\end{tabular}\vspace{0.5cm}
	\end{center}
	\vspace{-0.9cm}
\end{table}
The obtained upper bounds on the $\mathcal{L}_2$-gain are depicted in Fig.~\ref{Fig.L2GainTank}.\\\indent 
While the second order TP results in the smallest upper bound, combining two TPs of order two at $\omega_1=\begin{bmatrix}0.1 & 0.1 & 0\end{bmatrix}^T$ and $\omega_2=\begin{bmatrix}0.2 & 0.2 & 0\end{bmatrix}^T$ with a common storage function attains an even smaller estimation. Fig.~\ref{Fig.L2GainTank} also indicates that the effect of conservative bounds $M_{i,\alpha}$ is reduced by $E'(f)$ compared to $E(f)$. Indeed, $\bar{M}_{k+1}$ is significantly larger for combined TPs than for a single TP of the same order while still resulting in a meaningful bound on the $\mathcal{L}_2$-gain. \\\indent
Fig.~\ref{Fig.L2GainTank} also demonstrates that the nonlinear system can be approximated by a set of polynomials, which corresponds to $\beta_M=0$. Thus, a precise knowledge on $M_{i,\alpha}$ is not crucial as also a set of polynomials can explain the data but might amount to a more conservative inference.\\\indent
It might be surprising that the envelope for a second order TP yields a lower estimation than a third order TP. However, TPs only guarantee a reduction of the approximation error close to $\omega$ for increasing $k$ and the width of the polynomial sector bound \eqref{Envelop1} grows with $||\left(x-\omega\right)^\alpha||_2,|\alpha|=k+1$. Therefore, we observe a refined performance of the first and third order TP for a larger and smaller operation set, respectively.\\\indent
\begin{table}
	\begin{center}
		\caption{Computation time for Algorithm 1 and Theorem~\ref{ThmDissi}/Corollary~\ref{ThmDissi2}}\label{Computation}
		\begin{tabular}{c||c|c|c|c}
			$k$ & 1 & 2 & 2 (2 envelops) & 3 \\\hline
			Algorithm \ref{SVP}\phantom{\Big|} & 3.5\,s & 3.6\,s & 7.1\,s & 4.5\,s\\ 
			Theorem~\ref{ThmDissi}/Corollary~\ref{ThmDissi2}\phantom{\Big|} & 4.8\,s & 5.9\,s & 45\,s & 8.0\,s
		\end{tabular}
	\end{center}\vspace{-0.5cm}
\end{table}
Table~\ref{Computation} shows the computation times for obtaining the bounds on the partial derivatives and the noise by Algorithm~\ref{SVP} and for solving the SOS optimization problem in Theorem~\ref{ThmDissi} and Corollary~\ref{ThmDissi2} on a Lenovo i5 notebook.\\\indent
We compare our framework to a least-squares-error (LSE) model. To this end, we suppose that $f_i(h_i)=-c_i\sqrt{h_i}$ by first principles \cite{Quanser} and estimate the unknown coefficients $c_u$ and $c_i,i=1,2,$ from the available data using LSE estimation. To infer on the $\mathcal{L}_2$-gain of this nonlinear model, we retrieve the polynomial sector \eqref{PolySecBound} of the square root function by its third order TP and choosing $M_{1,[k+1\ 0\ 0]}=M_{2,[k+1\ 0\ 0]}$ and $M_{2,[0\ k+1\ 0]}$ such that the resulting polynomial sector bound contains the square root function for $h_{1},h_2\in[0.05, 0.25]$. Fig.~\ref{Fig.L2GainTank} illustrates that the calculated gain is smaller than the inferences from the data-driven TPs.\\\indent
To validate the $\mathcal{L}_2$-gain estimations from the LSE model and the data-driven TPs, we unfortunately do not know the actual $\mathcal{L}_2$-gain of the two-tank system due to the unknown system dynamics. Therefore, we calculate from the smoothed input-state trajectories the lower bound on the $\mathcal{L}_2$-gain
\begin{equation}\label{LowerBound}
\max_{\substack{\tau=4,\dots,149,\\ i=1,\dots,150-\tau}}\frac{\sqrt{\sum_{k=i}^{i+\tau}||\begin{bmatrix}\tilde{h}_{1k} & \tilde{h}_{2k}\end{bmatrix}^T||_2^2}}{\sqrt{\sum_{k=i}^{i+\tau}||\tilde{u}_k||_2^2}}.
\end{equation} 
Note that $\tau=1,2,3$ are neglected to average samples that lead to unrealistically large $\mathcal{L}_2$-gains due to the time delay of the inflow through the pump. Fig.~\ref{Fig.L2GainTank} clarifies that \eqref{LowerBound} is smaller than the obtained $\mathcal{L}_2$-gain inferences from the data-driven TPs, while the inference from the LSE model violates the lower bound \eqref{LowerBound}. This is the case even though the latter includes knowledge from first principles and is robust with respect to the nonlinearity due to the polynomial sector bound of the square root function. Therefore, we assess our approach better than the LSE estimation. Nevertheless, we cannot guarantee that the estimations from the data-driven TPs actually upper bound the true $\mathcal{L}_2$-gain as a guaranteed inference on the dynamics from only noisy data with unknown noise characterization is not possible.

\begin{rmk}[Learning continuous-time models]
	An alternative to obtain $E(f)$ for the continuous-time system \eqref{TrueSystem} by means of time derivatives of smoothed signals is to conclude on an envelope for a discrete-time representation directly from the measured input-state data. Then dissipativity of the continuous-time system could be analyzed by the continuous-time version of the discrete-time model under, e.g., Euler time-discretization. However, \cite{Filtervs} shows that this approach exhibits a deteriorated performance.
\end{rmk}  
\begin{rmk}[Numerical problems]
	We observed  numerical problems for the SOS optimization problem of Theorem~\ref{ThmDissi} and Corollary~\ref{ThmDissi2} if the state space $\mathbb{X}$ is not contained in the unit ball. Therefore, the water heights are given in meter.
\end{rmk}

\section{Data-driven incremental dissipativity verification}\label{SecIncr}

In this section, we apply the data-based representation for nonlinear functions from Section~\ref{SecNPModel} to verify incremental dissipativity for the unknown nonlinear continuous-time system \eqref{TrueSystem} operated within the compact set \eqref{Constraints} by the noisy measurements \eqref{DataSys}. While this pursues similarly to Section~\ref{SecDissi}, it shows that our TP representation is not limited to verify dissipativity but can easily be adapted to a wide range of problems in control theory.\\\indent
First, we recap and adapt results on incremental dissipativity from \cite{IncrementalRoland}.

\begin{defn}[Incremental $(Q,S,R)$-dissipativity]\label{DefIncDissi}
	System~\eqref{TrueSystem} is incremental dissipative on $\mathbb{P}$ \eqref{Constraints} with respect to the quadratic supply rate 
	\begin{equation}\label{IncSupply}
	s(x_1-x_2,u_1-u_2)=\begin{bmatrix} x_1-x_2\\u_1-u_2\end{bmatrix}^T\begin{bmatrix}Q & S\\ S^T & R\end{bmatrix}\begin{bmatrix} x_1-x_2\\ u_1-u_2\end{bmatrix}
	\end{equation}
	with $Q=Q^T\preceq0$ and $R=R^T$ if there exists a storage function $\lambda(x_1,x_2)=(x_1-x_2)^T\mathcal{M}(x_1,x_2)(x_1-x_2)$ with matrix $\mathcal{M}(x_1,x_2)$ being symmetric, differentiable, bounded, and positive definite for all $x_1,x_2\in\mathbb{X}$ such that
	\begin{align*}\label{IncdissipativityInqu}
	\dot{\lambda}(x_1,x_2)\leq s(x_1-x_2,u_1-u_2) 
	\end{align*}
	for all $(x_1,u_1),(x_2,u_2)\in{\mathbb{P}}$.
\end{defn}

Incremental dissipativity constitutes an extension of the notion of dissipativity in Definition~\ref{DissiDef} as dissipativity between two arbitrary trajectories is investigated. Therefore, incremental dissipativity implies dissipativity but not vice versa in general. Hence, incremental system properties establish a stronger condition on the system dynamics, and thus a less conservative controller design. Indeed, the set of systems satisfying a certain incremental property is smaller than the set of systems for the corresponding non-incremental property. Since the controller has to deal with all systems of these sets by design, the smaller set of systems for incremental properties leads to less restrictions for the controller design, and thus to a better closed-loop performance.\\\indent 
Incremental dissipativity is studied in \cite{IncrementalRoland} by the differential form of system \eqref{TrueSystem}  
\begin{equation}\label{DiffSys}
\delta \dot{x}=\frac{\partial f(\bar{x},\bar{u})}{\partial \xi}\delta \xi
\end{equation}
with $\xi=\begin{bmatrix}x^T&u^T\end{bmatrix}^T$ and $\delta\xi=\begin{bmatrix}\delta x^T&\delta u^T\end{bmatrix}^T$. Here, solutions $(\delta x(t),\delta u(t)), t\geq0,$ of the variational system \eqref{DiffSys} describe the time evolution of infinitesimal perturbations along the trajectory $(\bar{x}(t),\bar{u}(t)), t\geq0,$ of \eqref{TrueSystem}, which is described more precisely in Remark 1 of \cite{DifferentialSys}. With the introduction of the variational system, we proceed with the notion of differential dissipativity. 

\begin{defn}[Differential $(Q,S,R)$-dissipativity]\label{DefDiffDissi}
	System~\eqref{TrueSystem} is differential dissipative on $\mathbb{P}$ \eqref{Constraints} with respect to the quadratic supply rate 
	\begin{equation}\label{DiffdissipativityInqu}
	s(\delta x,\delta u)=\begin{bmatrix}\delta x\\\delta u\end{bmatrix}^T\begin{bmatrix}Q & S\\ S^T & R\end{bmatrix}\begin{bmatrix} \delta x\\\delta u\end{bmatrix}
	\end{equation}
	with $Q=Q^T\preceq0$ and $R=R^T$ if there exists a storage function $\lambda(x,\delta x)=\delta x^T\mathcal{M}(x)\delta x$ with matrix $\mathcal{M}(x)$ being symmetric, differentiable, bounded, and positive definite for all $x\in\mathbb{X}$ such that
	\begin{align}\label{DifdissipativityInqu}
	\dot{\lambda}(x,\delta x)\leq s(\delta x,\delta u) 
	\end{align}
	for all $(x,u)\in{\mathbb{P}}$ and all $(\delta x,\delta u)\in\mathbb{R}^{n_x}\times\mathbb{R}^{n_u}$.
\end{defn}

While we refer to \cite{IncrementalRoland} for an interpretation of differential dissipativity, we are interested in verifying incremental dissipativity by the subsequent connection of incremental and differential dissipativity.
\begin{thm}[Theorem 13 \cite{IncrementalRoland}]\label{ThmIncDiffDissi}
	System~\eqref{TrueSystem} is $(Q,S,R)$-incremental dissipative on $\mathbb{P}$ with respect to the incremental supply rate \eqref{IncSupply} if the system is $(Q,S,R)$-differential dissipative on $\mathbb{P}$ with respect to the differential supply rate \eqref{DiffdissipativityInqu}.	
\end{thm}
\begin{figure*}
	\setcounter{MYtempeqncnt}{\value{equation}}
	\setcounter{equation}{44}
	\begin{align}\label{SOSCondIncDissi}
	&\tilde{\Psi}(x,u)=\tilde{\Upsilon}^T
	\text{diag}\left(-\begin{bmatrix}0 & \mathcal{M}(x)\\ \mathcal{M}(x) & 0\end{bmatrix}\Bigg|-\frac{1}{2}\begin{bmatrix}0 & \left[\frac{\partial \mathcal{M}}{\partial x_1}\cdots\frac{\partial \mathcal{M}}{\partial x_{n_x}}\right]^T\\ \frac{\partial \mathcal{M}}{\partial x_1}\cdots\frac{\partial \mathcal{M}}{\partial x_{n_x}} & 0\end{bmatrix}\Bigg| \begin{bmatrix}Q & S\\ S^T & R\end{bmatrix}\Bigg|\sum_{i=1}^{n_p}T_i(x,u)p_i(x,u)\bigg|\dots\right.\notag\\
	&\hspace{1.45cm}\left.\tau_{\text{sec},0}(x,u)\tilde{\Phi}(\omega)\bigg|\tau_{\text{sm},0}(x,u)\varDelta_*\bigg|\tau_{\text{sec},1}(x,u)\Phi(\omega)\bigg|\cdots\bigg|\tau_{\text{sec},n_x}(x,u)\Phi(\omega)\bigg| \tau_{\text{sm},1}(x,u)\varDelta_*\bigg| \cdots\bigg| \tau_{\text{sm},n_x}(x,u)\varDelta_*
	\right)\tilde{\Upsilon}
	\end{align}
	\hrulefill
	\setcounter{equation}{\value{MYtempeqncnt}}
\end{figure*}

By Theorem~\ref{ThmIncDiffDissi}, we can verify incremental dissipativity by checking condition ~\eqref{DifdissipativityInqu}. For that purpose, elaborating the left-hand side of \eqref{DifdissipativityInqu} with the dynamics of the variational system \eqref{DiffSys} yields
\begin{align}\label{TimeDerivative}
\dot{\lambda}(x,\delta x)=&2 \delta x^T\mathcal{M}(x)\frac{\partial f({x},{u})}{\partial \xi} \delta\xi\notag\\
&+\delta x^T\left(\sum_{i=1}^{n_x}\frac{\partial \mathcal{M}({x})}{\partial x_i}f_i(x,u)\right)\delta x.	
\end{align}
Since $f$ and its first order partial derivative $\frac{\partial f({x},{u})}{\partial \xi}$ are unknown, we require a data-based envelope for $\frac{\partial f({x},{u})}{\partial \xi}\delta\xi$ and $\delta x_if(x,u),i=1,\dots,n_x$. For the latter, multiplying inequality \eqref{PolySecBound} for a TP $T_k(\omega)[f(x,u)]=A^*z(x,u)$ with $\delta x_i^2$ yields the envelope $\bar{E}(\delta x_if)$ equal to 
\begin{align}\label{Envelopdeltaf}
\bigcup_{A\in\tilde{\Sigma}_{\omega}}\left\{\left(\begin{bmatrix}\xi \\ \delta x_i\end{bmatrix},y\right)\in\mathbb{P}{\times}\mathbb{R}{\times}\mathbb{R}^{n_x}:\delta x_i^2p_\text{sec}\left(\xi,\frac{y}{\delta x_i},A\right){\leq}0 \right\}
\end{align}
which contains the graph of $\delta x_if$. To recover an envelope for $\frac{\partial f({x},{u})}{\partial \xi}\delta\xi$, we write the first order partial derivatives by Taylor's theorem as $\frac{\partial f_i({x},{u})}{\partial \xi_j}(x)=T_{k-1}(\omega)\left[\frac{\partial f_i({x},{u})}{\partial \xi_j}\right]+R_{k-1}(\omega)\left[\frac{\partial f_i({x},{u})}{\partial \xi_j}\right]$ with
\begin{align*}
T_{k-1}(\omega)\left[\frac{\partial f_i({x},{u})}{\partial \xi_j}\right]&=\sum_{|\alpha|=0}^{k-1}\frac{\partial^{|\alpha|+1}f_i(\omega)}{\partial \xi^{\alpha+e_j}}\frac{(\xi-\omega)^{\alpha}}{\alpha!}\\
&={a_i^*}^T\frac{\partial z(x,u)}{\partial \xi_j}
\end{align*}
and the Lagrange remainder
\begin{align*}
R_{k-1}(\omega)\left[\frac{\partial f_i({x},{u})}{\partial \xi_j}\right]{=}\sum_{|\alpha|=k}\frac{\partial^{k+1}f_i(\omega{+}\nu(\xi{-}\omega))}{\partial \xi^{\alpha+e_j}}\frac{\left(\xi-\omega\right)^{\alpha}}{\alpha !}
\end{align*}
which square can be upper bounded under Assumption~\ref{AssBoundDeri} by 
\begin{align}\label{BoundSquare}
&R^{\text{poly}}_{k-1}(\omega)\left[\frac{\partial f_i({x},{u})}{\partial \xi_j}\right]=\sum_{|\alpha|=k}\kappa_iM_{i,\alpha+e_j}^2\frac{\left(\xi-\omega\right)^{2\alpha}}{\alpha !^2}	
\end{align}
with $\kappa_i\in\mathbb{N}$ equal to the number of $M_{i,\alpha+e_j}\neq0$ with $|\alpha|=k$ by Lemma~\ref{LemBound}. Note that ${a_i^*}^T$ is the $i$-th row of the coefficient matrix ${A^*}$ and $z(x,u)$ is the vector of polynomials of TP $T_k(\omega)[f(x,u)]=A^*z(x,u)$. Summarizing,
\begin{equation*}
\frac{\partial f({x},{u})}{\partial \xi}\delta\xi=A^*\frac{\partial z(x,u)}{\partial \xi}\delta\xi+R_{k-1}(\omega)\left[\frac{\partial f({x},{u})}{\partial \xi}\right]\delta\xi
\end{equation*}
where the $(i,j)$-th element of $R_{k-1}(\omega)\left[\frac{\partial f({x},{u})}{\partial \xi}\right]$ is equal to $R_{k-1}(\omega)\left[\frac{\partial f_i({x},{u})}{\partial x_j}\right]$. With $R_j$ denoting the $j$-th column of $R_{k-1}(\omega)\left[\frac{\partial f({x},{u})}{\partial \xi}\right]$, we finally derive the polynomial sector bound
\begin{align*}
&\Bigg|\Bigg|\frac{\partial f({x},{u})}{\partial \xi}\delta\xi-A^*\frac{\partial z(x,u)}{\partial \xi}\delta\xi\Bigg|\Bigg|_2^2=\Bigg|\Bigg|\sum_{j=1}^{n_x+n_u}\delta\xi_jR_j\Bigg|\Bigg|_2^2\\
=&\sum_{i=1}^{n_x}\left(\sum_{j=1}^{n_x+n_u}\delta\xi_j R_{k-1}(\omega)\left[\frac{\partial f_i({x},{u})}{\partial \xi_j}\right] \right)^2\\
\leq&\sum_{i=1}^{n_x}\sum_{j=1}^{n_x+n_u}\pi_i\delta\xi_j^2 R^\text{poly}_{k-1}(\omega)\left[\frac{\partial f_i({x},{u})}{\partial \xi_j}\right]\\
=&\delta\xi^TR^{\text{poly}}_{k-1}(\omega)\delta\xi
\end{align*}
where $R^{\text{poly}}_{k-1}(\omega)=\sum_{i=1}^{n_x}\pi_i\,\text{diag}\left(R^{\text{poly}}_{k-1}(\omega)\left[\frac{\partial f_i({x},{u})}{\partial \xi_1}\right]\Big|\dots\Big|\right.$ $\left. R^{\text{poly}}_{k-1}(\omega)\left[\frac{\partial f_i({x},{u})}{\partial \xi_{n_x+n_u}}\right]\right)$ and $\pi_i,i=1,\dots,n_x,$ is equal to the number of $R^\text{poly}_{k-1}(\omega)\left[\frac{\partial f_i({x},{u})}{\partial \xi_j}\right]\neq0$ for $j=1,\dots,n_x+n_u$. The inequality follows from \eqref{BoundSquare} and the fact that $2vw\leq v^2+w^2,\forall v,w\in\mathbb{R}$, as exploited in Lemma~\ref{LemBound}. This yields the envelope 
\begin{align}\label{EnvelopPartialdel}
&\tilde{E}\left(\frac{\partial f({x},{u})}{\partial \xi}\delta\xi\right){=}\bigcup_{A\in\tilde{\Sigma}_{\omega}}\left\{\left(\begin{bmatrix}\xi \\ \delta\xi\end{bmatrix},y\right){\in}\mathbb{P}\times\mathbb{R}^{n_x+n_u}\times\mathbb{R}^{n_x}{:}\right.\notag\\
&\hspace{3cm}\left.\tilde{p}_\text{sec}\left(\begin{bmatrix}\xi \\ \delta\xi\end{bmatrix},y,A\right)\leq0 \right\}
\end{align}
with 
\begin{align*}
\tilde{p}_\text{sec}&=\star^T
\tilde{\Phi}(\omega)\cdot\begin{bmatrix}\begin{array}{c}\begin{matrix}
I_{n_x} & -I_{n_x} & 0 \\\hline 0 & 0 & I_{n_x+n_u}\end{matrix}\end{array}\end{bmatrix}
\begin{bmatrix}y\\ A\frac{\partial z(x,u)}{\partial \xi}\delta\xi\\\delta\xi	\end{bmatrix},\\
\tilde{\Phi}(\omega)&=\text{diag}\left(I_{n_x}\big| -R^{\text{poly}}_{k-1}(\omega)\right).
\end{align*}
In \eqref{EnvelopPartialdel}, we exploit the same set-membership $\tilde{\Sigma}_{\omega}$ for the TP $T_k(\omega)[f(x,u)]=A^*z(x,u)$ as in Theorem \ref{ThmDissi} because the TPs of the first order partial derivatives are also explained by the coefficient matrix $A^*$. Therefore, we can consider input-state-velocity data \eqref{DataSys} for verifying differential dissipativity as in the non-incremental case.

\begin{coro}[Data-driven increm. dissipativity verification]\label{CoroIncDissi}
	Let Assumption~\ref{AssBoundDeri} and \ref{AssNoiseBound} hold. At the same time, let the matrix $\tilde{Z}=\begin{bmatrix}z(\tilde{x}_1,\tilde{u}_1) & \cdots & z(\tilde{x}_S,\tilde{u}_S)\end{bmatrix}$ has full row rank for the TP $T_k(\omega)[f(x,u)]=A^*z(x,u)$ at $\omega=[\omega_x^T\ \omega_u^T]^T$. Then the set-membership 
	\begin{equation}\label{SetMem}
	\tilde{\Sigma}_{\omega}=\left\{A:\begin{bmatrix}I_{n_{z}}\\A\end{bmatrix}^T\varDelta_{*}\begin{bmatrix}I_{n_{z}}\\A\end{bmatrix}\preceq0\right\}
	\end{equation}
	for $A^*$ exist. Moreover, system~\eqref{TrueSystem} is incremental $(Q,S,R)$-dissipative on \eqref{Constraints} if there exist a symmetric matrix $\mathcal{M}(x)\in\text{SOS}[x]^{n_x\times n_x}$, matrices $T_i\in\text{SOS}[x,u]^{n_x+n_u},i=1,\dots,n_p$, and polynomials $\tau_{\text{sec},i},\tau_{\text{sm},i}\in\text{SOS}[x,u],i=0,\dots,n_x,$ such that $\tilde{\Psi}(x,u)$ in \eqref{SOSCondIncDissi} becomes an SOS matrix with $\tilde{\Upsilon}$ equal to
	\begin{equation*}
	\begin{bmatrix}\hspace{0.1cm}\begin{matrix}	 I_{n_x} & 0 & 0 & 0 & 0 \\ 0 & 0 & 0 & 0 & [I_{n_x} \ 0]\\	\hline 
	0 & I_{n_x^2} & 0 & 0 & 0 \\ 0 & 0 & 0 & 0 & [I_{n_x} \ 0]\\	\hline 	
	0 & 0 & 0 & 0 & I_{n_x+n_u}\\\hline 
	0 & 0 & 0 & 0 & I_{n_x+n_u}\\\hline 
	I_{n_x} & 0 & -I_{n_x} & 0 & 0  \\  0 & 0 & 0 & 0 & I_{n_x+n_u}\\\hline 
	0 & 0 & 0 & 0 & \frac{\partial z}{\partial \xi}\\ 0 & 0 & I_{n_x} & 0 & 0\\\hline
	0 & \text{diag}(e_1^T|\cdots|e_1^T)\phantom{\Big|}  & 0 & -[I_{n_x} \cdots 0]& 0  \\  0 & 0 & 0 & 0 & e_1^T\\\hline 
	\vdots & \vdots & \vdots & \vdots & \vdots \\\hline
	0 & \text{diag}(e_{n_x}^T|\cdots|e_{n_x}^T)\phantom{\Big|}  & 0 & -[0\cdots  I_{n_x}] & 0  \\  0 & 0 & 0 & 0 & e_{n_x}^T\\\hline
	0 & 0 & 0 & 0 & \phantom{\Big|}ze_1^T\\ 0 & 0 & 0 & [I_{n_x}   \cdots0] & 0\\\hline
	\vdots & \vdots & \vdots & \vdots & \vdots \\\hline
	0 & 0 & 0 & 0 & \phantom{\Big|}ze_{n_x}^T\\ 0 & 0 & 0 & [0\cdots I_{n_x}] & 0
	\end{matrix} \hspace{0.07cm}\end{bmatrix}\hspace{-0.11cm}.	
	\end{equation*}
	
	\setcounter{equation}{45}
\end{coro}
\begin{proof}
	Proposition~\ref{Primal_elli} implies immediately the existence of $\tilde{\Sigma}_{\omega}$. Furthermore, we pursue the proof of Theorem~\ref{ThmDissi}. Thus, $0\leq\sigma^T\tilde{\Psi}\sigma$ with 
	\begin{align*}
	\sigma=&\left[\begin{matrix}\left(\frac{\partial f}{\partial \xi} \delta\xi\right)^T &  \left[f_1\delta x^T \ \cdots \ f_{n_x}\delta x^T\right] &(A\frac{\partial z}{\partial \xi}\delta\xi)^T\end{matrix}\right.\dots\\
	&\left.\begin{matrix} & \left[\delta x_1(Az)^T \ \cdots \ \delta x_{n_x}(Az)^T\right] & \delta\xi^T\end{matrix} \phantom{\Big|^!}\right]^T
	\end{align*}
	and
	\begin{align*}
	\sigma^T\tilde{\Psi}\sigma=&\, -2 \delta x^T\mathcal{M}\frac{\partial f}{\partial \xi} \delta\xi-\delta x^T\left(\sum_{i=1}^{n_x}\frac{\partial \mathcal{M}}{\partial x_i}f_i\right)\delta x\\
	&+\delta \xi^T\begin{bmatrix}
	Q & S\\ S^T& R\end{bmatrix}\delta \xi+\sum_{i=1}^{n_P}\delta \xi^TT_i\delta \xi\ p_i\\	
	&+\tau_{\text{sec},0}\star^T
	\tilde{\Phi} \cdot\begin{bmatrix}\begin{array}{c}\begin{matrix}
	I_{n_x} & -I_{n_x} & 0 \\\hline 0 & 0 & I_{n_x+n_u}\end{matrix}\end{array}\end{bmatrix}\begin{bmatrix}\frac{\partial f}{\partial \xi} \delta\xi\\ A\frac{\partial z}{\partial \xi}\delta\xi\\\delta\xi	\end{bmatrix}\\
	&+\tau_{\text{sm},0}\star^T
	\varDelta_*\cdot
	\begin{bmatrix}I_{n_z}\\ A	\end{bmatrix}\frac{\partial z}{\partial \xi}\delta\xi\\
	&+\sum_{i=1}^{n_x}\tau_{\text{sec},i}\star^T
	{\Phi} \cdot\begin{bmatrix}\begin{array}{c}\begin{matrix}
	I_{n_x} & -I_{n_x} & 0 \\\hline 0 & 0 & 1\end{matrix}\end{array}\end{bmatrix}\begin{bmatrix}f\\ Az\\1	\end{bmatrix}\delta x_i\\
	&+\sum_{i=1}^{n_x}\tau_{\text{sm},i}\star^T
	\varDelta_*\cdot
	\begin{bmatrix}I_{n_z}\\ A	\end{bmatrix}\delta x_iz.
	\end{align*}
	Since $p_i(x,u)\leq0$ by \eqref{Constraints}, $\tilde{p}_\text{sec}\left(\begin{bmatrix}\xi \\ \delta\xi\end{bmatrix},y,A\right)\leq0$ by \eqref{EnvelopPartialdel}, $\delta x_i^2p_\text{sec}(\xi,\frac{y}{\delta x_i},A)\leq0$ by \eqref{Envelopdeltaf}, and $\star^T\varDelta_*\begin{bmatrix}I_{n_z}& A^T	\end{bmatrix}^T\preceq0$ by \eqref{SetMem}, the generalized S-procedure of Proposition~\ref{SOSRelaxation} implies 
	\begin{align}\label{IncDissiInqu}
	0\leq&\, -2 \delta x^T\mathcal{M}\frac{\partial f}{\partial \xi} \delta\xi-\delta x^T\left(\sum_{i=1}^{n_x}\frac{\partial \mathcal{M}}{\partial x_i}f_i\right)\delta x\notag\\
	&+\delta \xi^T\begin{bmatrix}
	Q & S\\ S^T& R\end{bmatrix}\delta \xi
	\end{align}
	for all $\left(\begin{bmatrix}\xi \\ \delta x_i\end{bmatrix},\delta x_if\right)\in \bar{E}(\delta x_if),i=1,\dots, n_x$, all $\left(\begin{bmatrix}\xi \\ \delta \xi\end{bmatrix},\frac{\partial f}{\partial \xi}\delta\xi\right)\in \tilde{E}\left(\frac{\partial f}{\partial \xi}\delta\xi\right)$, and all $\delta\xi\in\mathbb{R}^{n_x+n_u}$. Given that \eqref{IncDissiInqu} corresponds to \eqref{DifdissipativityInqu}  with \eqref{TimeDerivative} and the graphs of $\delta x_if$ and $\frac{\partial f}{\partial \xi}\delta\xi$ are contained in the envelopes $\bar{E}(\delta x_if)$ and $\tilde{E}\left(\frac{\partial f}{\partial \xi}\delta\xi\right)$, respectively, system~\eqref{TrueSystem} is $(Q,S,R)$-differential dissipative on $\mathbb{P}$, and therefore $(Q,S,R)$-incremental dissipative on $\mathbb{P}$ by Theorem~\ref{ThmIncDiffDissi}.
\end{proof}

Related to Theorem~\ref{ThmDissi}, the data-driven verification of incremental dissipativity relies on the knowledge of data-based envelopes that contain the unknown nonlinear system dynamics \eqref{Envelopdeltaf} and its Jacobian matrix \eqref{EnvelopPartialdel}. Since both envelopes require $\tilde{\Sigma}_\omega$, we can follow Algorithm~\ref{Algorithm2} to analyze incremental dissipativity but 
checking the SOS criterion for \eqref{SOSCondIncDissi} in step 3.). According to the discussion in Section~\ref{Sec_Improv}, we can reduce the conservatism of Corollary~\ref{CoroIncDissi} by incorporating prior knowledge on the structure of the dynamics. Furthermore, the combination of multiple envelopes from TPs around distinct points as in Corollary~\ref{ThmDissi2} is conceivable.

\subsection{Numerical example for incremental $\mathcal{L}_2$-gain}\label{SecEx3}

We apply Corollary~\ref{CoroIncDissi} to calculate the incremental $\mathcal{L}_2$-gain which is equivalent to the minimal $\gamma_i> 0$ such that the system is incremental dissipative for the supply rate $Q=-I_{n_x}, S=0$, and $R=\gamma_i^2I_{n_u}$. The system is motivated by the synchronization of two two-way coupled linear oscillators \cite{CoupledOsci}
\begin{align*}
\dot{x}_1 &= c_1x_1+c(x_1,x_2)+d_1,x_1\in[-0.7,0.7]\\
\dot{x}_2 &= c_2x_2-c(x_1,x_2)+c_uu+d_2,x_2\in[-0.7,0.7]
\end{align*}
with unknown coefficients $\begin{bmatrix}c_1 & c_2 & c_u\end{bmatrix}=\begin{bmatrix}-1 & -0.5 & 1\end{bmatrix}$, unknown nonlinear coupling $c(x_1,x_2)=0.3\sin(x_2-x_1)$, and unknown disturbances $d_{i},i=1,2$. Under the knowledge of the coefficients and the coupling, we determine the incremental $\mathcal{L}_2$-gain $1.93$ by Corollary~\ref{CoroIncDissi} for diagonal matrices $T_{i},i=1,2,$ and the fifth order TP of the system dynamics, i.e., $\tilde{\Sigma}_\omega=\{A^*\}$ for $\omega=0$. 
\begin{rmk}[Constant metric $\mathcal{M}$]
	We observe that if the system dynamics depend linearly on $u$, then a polynomial metric $\mathcal{M}(x)$ does not refine the result compared to a constant metric $\mathcal{M}$.
\end{rmk}

For the data-driven inference, we suppose the knowledge on the given structure above and $c(0,0)=0$. We draw $1000$ samples from a single trajectory with initial condition $x(0)=0$, input $u(t)=0.8\sin(0.03t^2+0.2t),t\geq0$, a sample rate of $0.1\,\text{s}$, and a bounded noise $||d_i(t)||_2\leq0.01||\dot{x}_{i}(t)||_2,i=1,2,t\geq0$. Additionally, we take $500$ snapshots into account, i.e., we evaluate the dynamics including noise at random states within $[-0.7,0.7]^2$ and zero input.\\\indent
Fig.~\ref{Fig.IncrL2} clarifies the influence of the bounds in Assumption~\ref{AssBoundDeri} on the data-driven inference on the incremental $\mathcal{L}_2$-gain obtained by Corollary~\ref{CoroIncDissi} with a constant metric $\mathcal{M}$ and TPs of order $5$.
\begin{figure}
	\centering
	\includegraphics[width=1\linewidth]{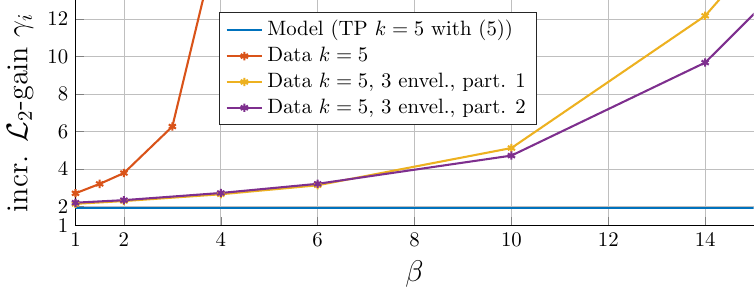}
	\caption{Data-based inference on the incremental $\mathcal{L}_2$-gain by TPs of order $5$.}
	\label{Fig.IncrL2}
\end{figure}
There, $M_{i,\alpha},i=1,2,|\alpha|=6,$ are given in view of the minimal upper bound $M^*_{i,\alpha}=0.3,i=1,2,|\alpha|=6$, i.e., $M_{i,\alpha}=\beta M^*_{i,\alpha}$.\\\indent
Since the inference by a single TP at $\omega=0$ is rather conservative and requires a precise knowledge of $M^*_{i,\alpha}$, we combine the local approximations of three envelopes to obtain the significantly improved conclusion of the third and fourth curve of Fig.~\ref{Fig.IncrL2}. To this end, we split for the former the state space into three regions $\mathbb{X}_1=[-0.7,-0.3]\times[-0.7,0.7],\mathbb{X}_2=[-0.3,0.3]\times[-0.7,0.7]$, and $\mathbb{X}_3=[0.3,0.7]\times[-0.7,0.7]$ and derive TPs at $\omega_1=\begin{bmatrix}-0.5 & 0 & 0\end{bmatrix}$, $\omega_2=0$, and $\omega_3=\begin{bmatrix}0.5 & 0 & 0\end{bmatrix}$. To reduce the computational effort, we exploit only the TP at $\omega_i$ to infer on the incremental property within subset $\mathbb{X}_i$. The common storage function ensures that the resulting incremental $\mathcal{L}_2$-gain holds for the whole operation set $[-0.7,0.7]^2$. We proceed analogously for the fourth curve but consider the partition regarding the $x_2$-coordinate and TPs at $\hat{\omega}_1=\begin{bmatrix}0 & -0.5 & 0\end{bmatrix}$, $\hat{\omega}_2=0$, and $\hat{\omega}_3=\begin{bmatrix}0 & 0.5 & 0\end{bmatrix}$.

\section{Conclusions}\label{SecConclusion}

We established a data-based characterization of nonlinear functions on the basis of Taylor polynomial approximation. To this end, a polynomial sector bound for nonlinear functions was received by means of Taylor's theorem and under a given bound on the magnitude of the $(k+1)$-th order partial derivatives. The latter can be deduced from data by a proposed validation procedure. Second, a set-membership for the unknown TP was developed based on noisy samples using ellipsoidal outer approximation. An extension of this approximation was proposed to incorporate prior knowledge on the nonlinear function. The combination of the polynomial sector bound and the set-membership achieved an envelope that is described by polynomials and contains the graph of the unknown nonlinear function. In an experimental example and a numerical example, we applied this envelope to determine non-incremental and incremental dissipativity properties solving an SOS optimization problem. Here a combination of multiple envelopes achieved significant improvements of the data-driven inference.\\\indent 
The presented set-membership perspective might stimulate further investigations of polynomial interpolation techniques for data-driven system analysis and control in future works.

\begin{IEEEbiography}[{\includegraphics[width=1.3in,height=1.25in,clip,keepaspectratio]{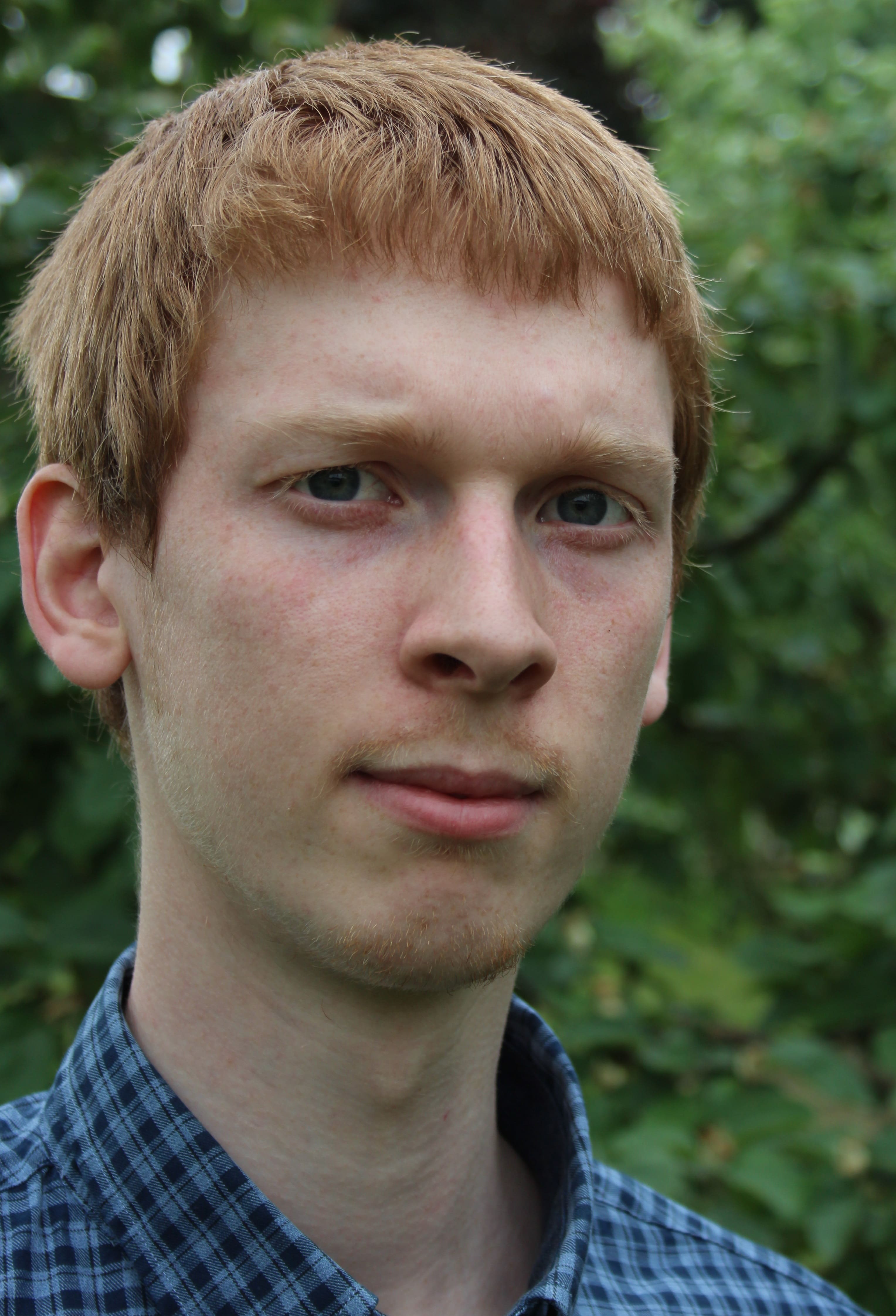}}] 
	{Tim Martin} (Graduate Student Member, IEEE) received the master’s degree in engineering cybernetics from the University of Stuttgart,
	Stuttgart, Germany, in 2018.\\\indent
	Since 2018, he has been a Research and Teaching Assistant with the Institute for Systems Theory and Automatic Control and a member of the Graduate School Simulation Technology, University of Stuttgart. His research interests include data-driven system analysis and control with focus on nonlinear systems.
\end{IEEEbiography}
\begin{IEEEbiography}[{\includegraphics[width=1.2in,height=1.25in,clip,keepaspectratio]{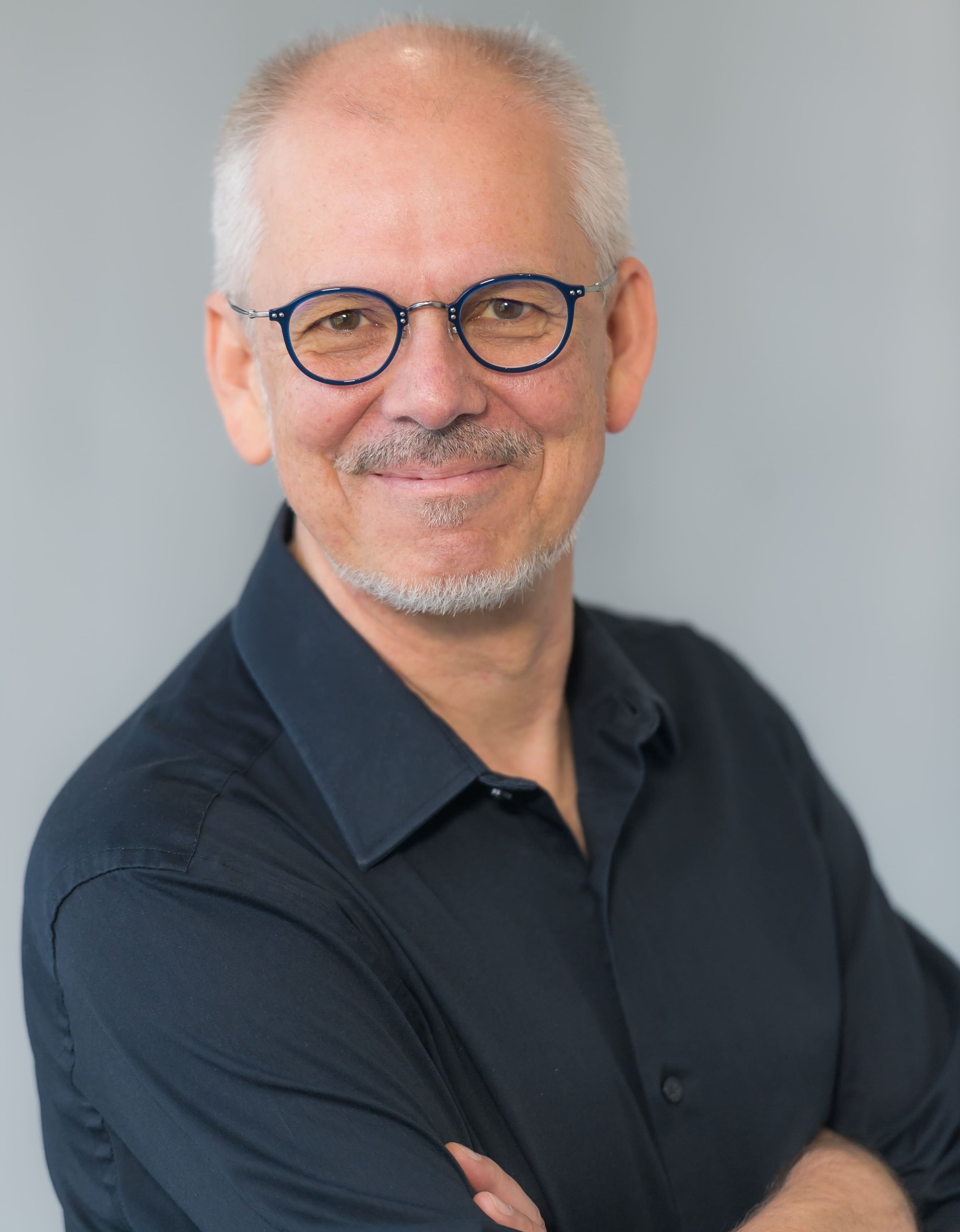}}]
	{Frank Allgöwer} (Member, IEEE) studied engineering cybernetics and applied mathematics in Stuttgart and with the University of California, Los Angeles (UCLA), CA, USA, respectively, and received the Ph.D. degree from the University of Stuttgart, Stuttgart, Germany.\\\indent
	Since 1999, he has been the Director of the Institute for Systems Theory and Automatic Control and a professor with the University of Stuttgart. His research interests include predictive control, data-based control, networked control, cooperative control, and nonlinear control with application to a wide range of fields including systems biology. \\\indent
	Dr. Allgöwer was the President of the International Federation of Automatic Control (IFAC) in 2017–2020 and the Vice President of the
	German Research Foundation DFG in 2012–2020.
\end{IEEEbiography}
\vfill

\end{document}